\documentclass[format=acmsmall, review=false, screen=true]{acmart}
\acmJournal{PACMPL}

\usepackage{booktabs}   

\usepackage{hyperref}
\usepackage{amsmath,amssymb,latexsym,amstext}
\usepackage{array}
\usepackage{url}
\usepackage{comment}
\usepackage{listings}
\usepackage{showexpl} 

\usepackage{enumitem}
\usepackage{booktabs}

\usepackage{stmaryrd}
\usepackage{xcolor}
\usepackage{mathtools}
\usepackage{apptools}
\usepackage{chngcntr}
\usepackage{scalerel}
\usepackage{xspace}
\usepackage{mathpartir}
\usepackage{graphicx}
\usepackage{caption}
\usepackage{csvsimple}
\usepackage{pdfpages}
\usepackage{tikz}
\usepackage{anyfontsize}
\usepackage{listings}

\def\withcolor{}

\ifdefined\withcolor
	\definecolor{haskellblue}{rgb}{0.0, 0.0, 1.0}
	\definecolor{gray_ulisses}{gray}{0.55}
	\definecolor{castanho_ulisses}{rgb}{0.0,0.4,0.0}
	\definecolor{preto_ulisses}{rgb}{0.41,0.20,0.04}
	\definecolor{green_ulisses}{rgb}{0.8,0.0,0.8}
\else
	\definecolor{haskellblue}{gray}{0.1}
	\definecolor{haskellred}{gray}{0.1}
	\definecolor{gray_ulisses}{gray}{0.1}
	\definecolor{castanho_ulisses}{gray}{0.1}
	\definecolor{preto_ulisses}{gray}{0.1}
	\definecolor{green_ulisses}{gray}{0.1}
\fi

\def\incodesize{\normalsize}
\def\codesize{\normalsize}

\lstdefinelanguage{HaskellUlisses} {
	basicstyle=\ttfamily\codesize,
	sensitive=true,
	mathescape=true,
	morecomment=[l][\color{gray_ulisses}\ttfamily\codesize]{--},
	morestring=[b]",
	stringstyle=\color{haskellred},
	showstringspaces=false,
	numberstyle=\codesize,
	numberblanklines=true,
	showspaces=false,
	breaklines=true,
	showtabs=false,
	emph=
	{[1]
		FilePath,IOError,abs,acos,acosh,all,and,any,appendFile,approxRational,asTypeOf,asin,
		asinh,atan,atan2,atanh,basicIORun,break,catch,ceiling,chr,compare,concat,concatMap,
		const,cos,cosh,curry,cycle,decodeFloat,denominator,digitToInt,div,divMod,drop,
		dropWhile,either,elem,encodeFloat,enumFrom,enumFromThen,enumFromThenTo,enumFromTo,
		error,even,exp,exponent,fail,filter,flip,floatDigits,floatRadix,floatRange,floor,
		fmap,foldl,foldl1,foldr,foldr1,fromDouble,fromEnum,fromInt,fromInteger,
		fromRational,fst,gcd,getChar,getContents,getLine,hd,tl,head,id,inRange,index,init,intToDigit,
		interact,ioError,isAlpha,isAlphaNum,isAscii,isControl,isDenormalized,isDigit,isHexDigit,
		isIEEE,isInfinite,isLower,isNaN,isNegativeZero,isOctDigit,isPrint,isSpace,isUpper,iterate,
		last,lcm,length,lex,lexDigits,lexLitChar,lines,log,logBase,lookup,map,mapM,mapM_,max,
		maxBound,maximum,maybe,min,minBound,minimum,mod,negate,not,notElem,numerator,odd,
		or,pi,pred,primExitWith,print,product,properFraction,putChar,putStr,putStrLn,quot,
		quotRem,range,rangeSize,read,readDec,readFile,readFloat,readHex,readIO,readInt,readList,readLitChar,
		readLn,readOct,readParen,readSigned,reads,readsPrec,realToFrac,recip,rem,repeat,replicate,
		reverse,round,scaleFloat,scanl,scanl1,scanr,scanr1,seq,sequence,sequence_,show,showChar,showInt,
		showList,showLitChar,showParen,showSigned,showString,shows,showsPrec,significand,signum,sin,
		sinh,snd,span,splitAt,sqrt,subtract,succ,sum,tail,take,takeWhile,tan,tanh,threadToIOResult,toEnum,
		toInt,toInteger,toLower,toRational,toUpper,truncate,uncurry,undefined,unlines,until,unwords,unzip,
		unzip3,userError,words,writeFile,zip,zip3,zipWith,zipWith3,listArray,doParse,for,initTo,
        maxEvens,create,get,set,initialize,idVec,fastFib,fibMemo,
        insert,union,split,size,fromList,initUpto,trim,quickSort,insertSort,append,upperCase,
        copy, group, doDownLoop, mapAccumR, peekByteOff,
        pokeByteOff,spanByte,
        good, bad, foo, explode,
        fib, ack,
        tLen,
        memcpy,writeChar,unsafeWrite,unsafeFreeze,
        singleton,
				ex1,ex2,ex3,ex4,incr,inc,dec,compose
	},
	emphstyle={[1]\color{haskellblue}},
	emph=
	{[2]
		Bool,Char,Double,Either,Float,IO,Integer,Int,Maybe,Ordering,Rational,Ratio,ReadS,ShowS,String,
		Word8,Nat,NonZero,Nat64,Text,ByteString,ByteStringSZ,ByteStringN,
    Ptr,ForeignPtr,CSize
    InPacket,Tree,Prop,TreeEq,TreeLt,Vec,
    NullTerm,IncrList,DecrList,UniqList,BST,MinHeap,MaxHeap,
    PtrN,ByteStringN,ByteStringEq,VO,ByteStringsEq,ByteStringNE,
		List,Even
	},
	emphstyle={[2]\color{castanho_ulisses}},
	emph=
	{[3]
		case,class,data,deriving,do,else,if,return,def,import,in,infixl,infixr,instance,let,
		requires,ensures,assume,val
		module,measure,predicate,of,primitive,then,refinement,type,where,lazy
	},
	emphstyle={[3]\color{preto_ulisses}\textbf},
	emph=
	{[4]
		quot,rem,div,mod,elem,notElem,seq
	},
	emphstyle={[4]\color{castanho_ulisses}\textbf},
	emph=
	{[5]
		PS,Tip,Node,EQ,False,GT,Just,LT,Left,Nothing,Right,True,Show,Eq,Ord,Num,
		Cons,Nil
	},
	emphstyle={[5]\color{green_ulisses}}
}

\lstnewenvironment{fscode}
{\lstset{language=HaskellUlisses,basicstyle=\ttfamily\small}}
{}

\lstnewenvironment{code}
{\lstset{language=HaskellUlisses}}
{}

\lstnewenvironment{codebox}
{\lstset{language=HaskellUlisses,frame=tlbr,basicstyle=\ttfamily\codesize}}
{}

\lstMakeShortInline[language=HaskellUlisses,basicstyle=\ttfamily\incodesize]@

\usepackage{commands}

\usepackage[english]{babel}

\makeatletter\if@ACM@journal\makeatother
\acmJournal{PACMPL}
\acmVolume{1}
\acmNumber{1}
\acmArticle{26}
\acmYear{2017}
\acmMonth{9}
\acmDOI{10.1145/3110270}
\acmPrice{}
\startPage{1}
\else\makeatother
\acmConference[PL'17]{ACM SIGPLAN Conference on Programming Languages}{January 01--03, 2017}{New York, NY, USA}
\acmYear{2017}
\acmISBN{978-x-xxxx-xxxx-x/YY/MM}
\acmDOI{10.1145/3110270}
\startPage{1}
\fi

\begin{document}

\title[]{Local Refinement Typing} 


\acmBadgeR[http://icfp17.sigplan.org/track/icfp-2017-Artifacts]{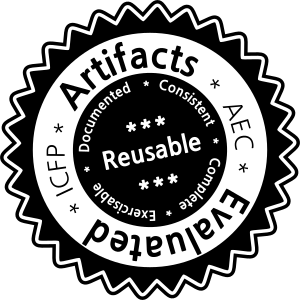}

\author{Benjamin Cosman}
\affiliation{\institution{UC San Diego} \country{USA}}
\email{blcosman@cs.ucsd.edu}

\author{Ranjit Jhala}
\affiliation{\institution{UC San Diego} \country{USA}}
\email{jhala@cs.ucsd.edu}




\begin{abstract}
We introduce the \toolname algorithm for local
refinement type inference, yielding a new SMT-based
method for verifying programs with polymorphic data types
and higher-order functions.
\toolname is \emph{concise} as the programmer need
only write signatures for (externally exported)
top-level functions and places with cyclic (recursive)
dependencies, after which \toolname can \emph{predictably}
synthesize the \emph{most precise} refinement types
for all intermediate terms (expressible in the
decidable refinement logic), thereby checking
the program without false alarms.
We have implemented \toolname and evaluated it on the
benchmarks from the \toolfont{LiquidHaskell} suite
totalling about 12KLOC.
\toolname checks an existing safety benchmark
suite using about half as many templates
as previously required and nearly $2\times$ faster.
In a new set of theorem proving benchmarks
\toolname is both $10-50 \times$ faster and,
by synthesizing the most precise types,
avoids false alarms to make verification possible.
\end{abstract}

%


\setcopyright{rightsretained}


\maketitle

\section{Introduction}\label{sec:intro}

Refinement types are a generalization of
Floyd-Hoare logics to the higher-order
setting.
Assertions are generalized to types by
constraining basic types so that their
inhabitants satisfy predicates drawn from
SMT-decidable refinement logics.
Assertion checking can then be generalized
to a form of subtyping, where subtyping for
basic types reduces to checking implications
between their refinement predicates~\cite{Constable86,Rushby98}.
Thus, thanks to the SMT revolution~\cite{SMTLIB2},
refinement types have emerged as a modular and
programmer-extensible means of expressing and
verifying properties of polymorphic, higher order
programs in languages like
{ML}~\cite{pfenningxi98,Dunfield07,LiquidPLDI08},
{Haskell}~\cite{LiquidICFP14},
{Racket}~\cite{Kent16},
{$F^\sharp$}~\cite{GordonRefinement09}, and
{TypeScript}~\cite{Vekris16}.

\mypara{The Problem}
Unfortunately, the expressiveness and
extensibility offered by predicate
refinements comes at a price.
Existing refinement type systems are either
\emph{not concise},
\ie require many type annotations
(not just for top-level functions),
or
\emph{not complete},
\ie programs fail to type check because
the checker cannot synthesize suitable
types for intermediate terms,
or
\emph{not terminating},
\ie the checker can diverge while trying
to synthesize suitable types for intermediate
terms in an iterative counterexample-guided manner.
%
%
We show  that the presence of logical predicates
relating \emph{multiple} program variables renders
classical approaches like unification or set-constraint
based subtyping -- even when used in
a \emph{bidirectional}~\cite{pierce-turner} or
\emph{local} fashion~\cite{odersky-local-01}
-- not up to the task of inference.
The problem is especially acute for programs
that make ubiquitous use of polymorphic datatypes
and higher-order combinators.
Consequently, to use refinement types, the programmer
must accept one of several distasteful choices: litter
her code with type annotations (everywhere), or eschew
polymorphic combinators (\S~\ref{sec:overview}).

\mypara{\toolname: Local Refinement Typing}
In this paper, we introduce the
\toolname algorithm that permits
local refinement type inference.
\toolname is \emph{concise} as
the programmer need only write
signatures for (externally exported)
top-level functions, and places with
cyclic dependencies on refinements,
either explicit (\eg recursive functions) or
implicit (\eg calls to inductive
library functions like \kw{fold}).
\toolname then synthesizes the most
\emph{precise} refinement types for
all intermediate terms that are
expressible in the refinement logic,
and hence checks the specified signatures.
Finally, \toolname ensures that
verification remains \emph{decidable}
by using the framework of
Liquid Typing~\cite{LiquidPLDI08}
to synthesize refinements from
user-defined templates (``qualifiers'')
in the presence of cyclic dependencies.
\toolname reconciles concision, precision
and decidability with three contributions.

\begin{enumerate}

\item \textbf{From Programs to NNF Horn Constraints}
Our first insight is that we can phrase the
problem of refinement type checking as that
of checking the satisfiability of (nested)
Horn Clause Constraints in Negation Normal
Form (NNF).
The NNF constraints are \emph{implications}
between refinement variables denoting the unknown,
to-be-synthesized refinements; a satisfying assignment
for these constraints yields refinement types for
all intermediate terms, that satisfy the subtyping
obligations that must hold for the program to
meet its specification.
Crucially, our novel NNF constraint formulation
retains the \emph{scoping structure} of bindings
that makes precise constraint solving
practical~(\S~\ref{sec:constraints}).

\item \textbf{Solving Constraints via Elimination}
Our second insight is that we can find
a satisfying assignment for the NNF
constraint, by systematically computing
the \emph{most precise} assignment for
each refinement variable, and using it
to \emph{eliminate} the variable from
the constraints.
Unfortunately, we find that a direct
application of the elimination algorithm,
inspired by the classical notion of ``unfolding''
from the Logic Programming literature
\cite{burstall-darlington-77,TamakiSato84},
leads to an exponential blowup even for
simple programs that arise in practice.
We show that such a blowup cannot be avoided
\emph{in general} as the problem of refinement
type checking is \exptime-hard.
Fortunately, we show how to exploit the
scoping structure present in NNF constraints
to get compact solutions for real-world
programs~(\S~\ref{sec:algorithm}).

\item\textbf{Refinement Typing: Faster and without False Alarms}
We have implemented \toolname and evaluated its
effectiveness on all the benchmarks from the
\toolfont{LiquidHaskell} suite~\cite{LiquidICFP14}.
We show that in \emph{existing} safety benchmarks
where qualifiers \emph{were} supplied or extracted
from top-level type signatures, \toolname is able
to synthesize the types about $2 \times$ faster
and requires only about half the qualifiers (needed
\eg to synthesize types for non-top-level recursive
functions).
Further, in a \emph{new} set of theorem proving
benchmarks which make heavy use of polymorphic
higher-order functions and data types, global
template-based inference of \toolfont{LiquidHaskell}
is both prohibitively slow \emph{and} unpredictable,
failing with false alarms.
%
%
In contrast, \toolname is more than $10 \times$ faster and,
by synthesizing the most precise types,
makes verification possible (\S~\ref{sec:eval}).

\end{enumerate}

\section{Overview}\label{sec:overview}

We start with a high-level overview illustrating
how refinement types can be used to verify programs,
the problems that refinements pose for local inference,
and our techniques for addressing them.
We have picked obviously contrived examples
that distill the problem down to its essence:
for real-world programs we refer the reader
to \cite{realworldliquid,fstar}.

\mypara{Refinement Types} We can precisely describe
\emph{subsets} of values corresponding to a type
by composing basic types with \emph{logical predicates}
that are satisfied by values inhabiting the type.
For example, \kw{Nat} describes the subset of \kw{Int}
comprising (non-negative) natural numbers:
\begin{code}
  type Nat = {v:Int | 0 <= v}
\end{code}

\mypara{Verification Conditions}
A refinement type checker can use the above signature to
check:
\begin{code}
  abs :: Int -> Nat
  abs n | 0 <= n    = n
        | otherwise = 0 - n
\end{code}
by generating a \emph{verification condition} (VC)~\cite{NelsonOppen}:
\begin{alignat*}{4}
\csbind{n}{} & 0 \leq n      & \ \Rightarrow\ & \csbind{\vv}{\vv = n}     & \ \Rightarrow 0 \leq \vv \\
    \wedge\  & 0 \not \leq n & \ \Rightarrow\ & \csbind{\vv}{\vv = 0 - n} & \ \Rightarrow 0 \leq \vv
\end{alignat*}
whose \emph{validity} can be checked by an SMT
solver~\cite{NelsonOppen,SMTLIB2} to ensure that
the program meets the given specification.

\subsection{The Problem}

For simple, first-order programs like \kw{abs},
VC generation is analogous to the classical
Floyd-Hoare approach used \eg by
\textsc{EscJava}~\cite{ESCJava}.
Namely, we have constraints \emph{assuming}
the path conditions (and pre-conditions),
and \emph{asserting} the post-conditions.
Next, we illustrate how VC generation
and checking is problematic in the presence
of \emph{local variables}, \emph{collections},
and \emph{polymorphic}, \emph{higher order}
functions.
We do so with a series of small but idiomatic
examples comprising ``straight-line'' code
(no explicit looping or recursion), decorated
with top-level specifications as required for
local type inference. We show how the presence
of the above features necessitates the synthesis
of intermediate refinement types, which is beyond
the scope of existing approaches.

\mypara{Library Functions}
Our examples use a library of
functions with types as shown
in Fig.~\ref{fig:lib}.

\begin{figure}[t]
  \centering
\begin{codebox}
inc  :: x:Int -> {v:Int | v = x + 1}   -- increment by one
dec  :: x:Int -> {v:Int | v = x - 1}   -- decrement by one
[]   :: List a                         -- list : "nil"
(:)  :: a -> $\tlist{\kw{a}}$ -> $\tlist{\kw{a}}$                      -- list : "cons"
last :: List a -> a                    -- last element of a list
map  :: (a -> b) -> $\tlist{\kw{a}}$ -> $\tlist{\kw{b}}$               -- mapping over a list
(.)  :: (b -> c) -> (a -> b) -> a -> c -- function composition
\end{codebox}
\caption{Type specifications for library functions}
\label{fig:lib}
\end{figure}

\mypara{Example 1: Local Variables}
In \kw{ex1} in Figure~\ref{fig:ex1},
we illustrate the problem posed
by local binders that lack (refinement)
type signatures.
Here, we need to check that \kw{inc y}
is a \kw{Nat} \ie non-negative, assuming
the input \kw{x} is.
To do so, we must synthesize a
sufficiently strong type for the
\emph{local} binder \kw{y} which says
that its value is an \emph{almost}-\kw{Nat},
\ie that \kw{y} has the type
$\reftpv{\kw{y}}{\tint}{0 \leq \kw{y} + 1}$.

\mypara{Example 2: Collections}
In \kw{ex2} in Figure~\ref{fig:ex2}
we show how the idiom of creating
intermediate collections of values
complicates refinement type checking.
Here, we use a \kw{Nat} to create a \emph{list}
\kw{ys} whose \kw{last} element is extracted
and incremented, yielding a \kw{Nat}.
The challenge is to infer that the intermediate
collection \kw{ys} is a list of almost-\kw{Nat}s.

\mypara{Example 3: Higher Order Functions}
Function \kw{ex3} shows the idiom of composing
two functions with the \kw{(.)} operator whose
type is in Figure~\ref{fig:lib}.
To verify the type specification for \kw{ex3}
we need to synthesize suitable instantiations
for the type variables \kw{a}, \kw{b} and \kw{c}.
Unlike with classical types, unification is
insufficient, as while \kw{a} and \kw{c} may be
instantiated to \kw{Nat}, we need to infer
that \kw{b} must be instantiated with
almost-\kw{Nat}.

\mypara{Example 4: Iteration}
Finally \kw{ex4} shows how the higher-order
\kw{map} and \kw{(.)} are used to idiomatically
refactor transformations of collections
in a ``wholemeal'' style~\cite{Hinze09}.
This straight-line function (treating \kw{map} as
a library function, abstracted by
its type signature) vexes existing refinement
type systems~\cite{Knowles10,fstar} that cannot
infer that \kw{map dec} (resp. \kw{map inc})
transforms a \kw{Nat} list (resp. almost-\kw{Nat} list)
into an almost-\kw{Nat} list (resp. \kw{Nat} list).

\subsection{Existing Approaches}

Existing tools use one of three approaches,
each of which is stymied by programs like
the above.

\mypara{1. Existentials}
First, \cite{Knowles09,Kent16} shows how to use
existentials to hide the types of local
binders.
In \kw{ex1}, this yields
$$\kw{y} \ \dcolon\ \exbind{0 \leq \kw{t}}{}\
                       \reftpv{\vv}{\tint}{\vv = \kw{t} - 1}$$
As \kw{inc y} has the type $\reft{\vv}{\vv = \kw{y} + 1}$ we
get the VC:
$$
\csbind{\kw{y}}{(\exbind{\kw{t}}{0 \leq \kw{t}} \wedge \kw{y} = \kw{t} - 1)} \ \Rightarrow\
\csbind{\vv}{\vv = \kw{y} + 1} \ \Rightarrow 0 \leq \vv
$$
which is proved valid by an SMT solver.
\sage~\cite{Knowles10}, \fstar~\cite{fstar}
and \textsc{Racket}~\cite{Kent16} use
this method to check \kw{ex1}.
However, it is not known how to scale this method
beyond local variable hiding, \ie to account for
collections, lambdas, higher-order functions \etc.
Consequently, the above systems \emph{fail to check}
\kw{ex2}, \kw{ex3} and \kw{ex4}.

\mypara{2. Counterexample-Guided Refinement}
Second, as in \textsc{Mochi} \cite{TerauchiPOPL13}
we can use \emph{counterexample-guided abstraction-refinement}
(CEGAR) to iteratively compute stronger refinements
until the property is verified.
This approach has two drawbacks.
First, it is non-modular in that it requires
\emph{closed} programs -- \eg the source
of library functions like \kw{map} and \kw{(.)} --
in order to get counterexamples.
Second, more importantly, it is limited to
logical fragments like linear arithmetic
where Craig Interpolation based
``predicate discovery'' is predictable,
and is notoriously prone to divergence
otherwise~\cite{jhalamcmillan06}.
Consequently, the method has only been
applied to small but tricky programs
comprising tens of lines, but not scaled
to large real-world libraries spanning
thousands of lines, and which typically
require refinements over uninterpreted
function symbols or
modular arithmetic~\cite{realworldliquid}.

\mypara{3. Abstract Interpretation}
Finally, the framework of Liquid Typing~\cite{LiquidPLDI08}
shows how to synthesize suitable refinement types
via abstract interpretation. The key idea is to:
\begin{enumerate}
\item \emphbf{Represent} types
   as \emph{templates} that use
   $\kvar$ \emph{variables}
   to represent unknown refinements.
\item \emphbf{Generate} subtyping
   constraints over the templates,
   that reduce to implications
   between the $\kvar$ variables,
\item \emphbf{Solve} the implications
   over an \emph{abstract domain},
   \eg the lattice of formulas
   generated by conjunctions of
   user-supplied atomic predicates,
   to find an assignment for the
   $\kvar$-variables that satisfies
   the constraints.
\end{enumerate}

Unfortunately, this approach also has several limitations.
First, the user must supply atomic predicates
(or more generally, a suitable abstract domain)
which should be superfluous in ``straight line''
code as in \kw{ex1}, \kw{ex2}, \kw{ex3}, and \kw{ex4}.
Unfortunately, without the atomic predicate hints,
Liquid Typing fails to check any of the above examples.
Specifically, unless the user provides the qualifier
or template ${0 \leq \kw{v} + 1}$ (which does not appear
anywhere in the code or specifications), the system cannot
synthesise the almost-\kw{Nat} types for the various
intermediate terms in the above examples.
Consequently \toolfont{LiquidHaskell} (which
implements Liquid Typing) rejects safe
(and checkable) programs, leaving the user with the
unpleasant task of debugging false alarms.
Second, in predicate abstraction, the so-called
abstraction operator ``$\alpha$'' makes \emph{many}
expensive SMT queries, which can slow down verification,
or render it impossible when the specifications, and
hence predicates, are complicated~(\S~\ref{sec:eval}).

\begin{figure}[t!]
\centering
\begin{minipage}[t]{0.30\textwidth}
\begin{codebox}
ex1 :: Nat -> Nat
ex1 x =
  let y =
    let t = x
    in
      dec t
  in
    inc y
$\wwedge$
\end{codebox}
\end{minipage}
\hspace{0.1in}
\begin{minipage}[t]{0.45\textwidth}
\begin{codebox}
ex2 :: Nat -> Nat
ex2 x =
  let ys = let n  = dec x
               p  = inc x
               xs = n : []
           in p : xs
      y  = last ys
  in
    inc y

\end{codebox}
\end{minipage}
\caption{Examples: (\bkw{ex1}) Local variables, (\bkw{ex2}) Collections}
\label{fig:ex1}
\label{fig:ex2}
\end{figure}

\subsection{Our Solution: Refinement \toolname}

Next, we describe our \toolname algorithm for
\emph{local refinement typing} and show how it
allows us to automatically, predictably and
efficiently check refinement types in the
presence of variable hiding, collections,
lambdas and polymorphic, higher order functions.
The key insight is two-fold.
First, we present a novel reduction from
type inference to the checking the
satisfaction of a system of NNF Horn
Constraints (over refinement variables)
that preserve the scoping structure of
bindings in the original source program.
Second, we show how to check satisfaction
of the NNF constraints by exploiting the
scoping structure to compute the strongest
possible solutions for the refinement
variables.
The above steps yield a method
that is \emph{concise}, \ie requires
no intermediate signatures, and
\emph{complete}, \ie is guaranteed to
check a program whenever suitable
signatures exist, instead of failing
with false alarms when suitable
qualifiers are not provided.
Finally, the method \emph{terminates},
as the acyclic variables can be eliminated
(by replacing them with their most precise
solution), after which any remaining cyclic
variables can be solved via abstract
interpretation~\cite{LiquidPLDI08}.

\subsubsection*{Example 1: Local Variables}
First, let us see how \toolname synthesizes
types for local binders that lack (refinement)
type signatures.


\mypara{1. Templates}
\toolname starts by generating templates for
terms whose type must be synthesized.
From the input type of \kw{ex1} (Fig~\ref{fig:ex1}),
the fact that \kw{t} is bound to \kw{x}, and the output
type of \kw{dec} we have:
\begin{align*}
\kw{x} & \ \dcolon\ \reftpv{\vv}{\tint}{0 \leq \vv} \\
\kw{t}     & \ \dcolon\ \reftpv{\vv}{\tint}{\vv = \kw{x}} \\
\kw{dec t} & \ \dcolon\ \reftpv{\vv}{\tint}{\vv = \kw{t} - 1} \\
\intertext{The term \kw{dec t} is bound to \kw{y}
but as \kw{t} goes out of scope we assign \kw{y}
a template}
\kw{y} & \ \dcolon\ {\reftpv{\vv}{\tint}{\kva{\kvar}{\vv}}} \\
\intertext{with a fresh $\kvar$ denoting the (unknown)
refinement for the expression. That is, \kw{y} is a
$\tint$ value $\vv$ that satisfies $\kva{\kvar}{\vv}$.
We will constrain $\kvar$ to be a well-scoped super-type
of its body \kw{dec t}. Finally, the output of \kw{ex1}
gets assigned}
\kw{inc y} & \ \dcolon\ \reftpv{\vv}{\tint}{\vv = \kw{y} + 1}
\end{align*}

\mypara{2. Constraints}
Next, \toolname generates constraints between the various
refinements. At each application, the parameter must be a
subtype of the input type; dually, in each function definition,
the body must be a subtype of the output type, and at each let-bind
the body must be a subtype of the whole let-in expression.
For base types, subtyping is exactly implication between
the refinement formulas \cite{Constable86,Rushby98}.
For \kw{ex1} we get two constraints.
The first relates the body of the let-in with its template;
the second the body of \kw{ex1} with the specified output:
\begin{alignat}{2}
\csbind{\kw{x}}{0 \leq \kw{x}} \ \Rightarrow\
  &                                                     && \csbind{\vv}{\vv = \kw{x} - 1}  \ \Rightarrow\ \kva{\kvar}{\vv}
  \label{eq:ex1:1}\\
\wedge\ \
  & \csbind{\kw{y}}{\kva{\kvar}{\kw{y}}} \ \Rightarrow\ && \csbind{\vv}{\vv = \kw{y} + 1}\ \Rightarrow\ 0 \leq \vv
  \label{eq:ex1:2}
\end{alignat}

\toolname ensures that the constraints preserve the
\emph{scoping structure} of bindings.
Each subtyping (\ie implication) happens under a
context where the program variables in scope are
bound and required to satisfy the refinements from
their previously computed templates.
Further, shared binders are explicit in the NNF
constraint. For example, in the NNF constraint
above, the shared binder \kw{x} in the source
program is also shared across the two implications
\ref{eq:ex1:1} and \ref{eq:ex1:2}.
This sharing crucially allows \toolname to
compute the most precise solution.

\mypara{3. Solution}
To \emph{solve} the constraints  we need to find a
\emph{well-formed assignment} mapping each $\kvar$-variable
to a predicate over variables in scope wherever
the $\kvar$-variable  appears, such that the
formula obtained by replacing the $\kvar$ with
its assignment is \emph{valid}.
If no such interpretation exists, then the constraints
are \emph{unsatisfiable} and the program is ill-typed.

\toolname computes an interpretation called
the \emph{strongest refinements} using a
method inspired by the classical notion
of ``unfolding'' from the Logic Programming
literature~\cite{burstall-darlington-77,TamakiSato84}.
In essence, if we have implications of the form:
  $P_i \Rightarrow \kva{\kvar}{\vv} $
then we can assign $\kvar$ to the disjunction
  $\kvapp{\kvar}{x} \doteq \vee_i P_i$
after taking care to existentially quantify variables.
In our example, $\kvar$ is a unary refinement predicate,
and so we use (\ref{eq:ex1:1}) to obtain the assignment
\begin{align}
\kva{\kvar}{z} & \ \doteq\ \exbind{x}{0 \leq x \ \wedge (\exbind{\vv}{\vv = x - 1 \ \wedge\ \vv = z})} \label{eq:ex1:sol}
\end{align}
which simplifies to ${0 \leq z + 1}$.
The computed assignment above is simply the
(existentially quantified) \emph{hypotheses}
arising in the conjunction where $\kva{\kvar}{\cdot}$
is the goal.
The ``simplification'' done above and in the sequel
is purely for \emph{exposition}: to \emph{compute} it
we would have to resort to quantifier elimination
procedures which are prohibitively expensive in
practice, and hence avoided by \toolname.

It is easy to check that the above assignment renders
the constraint valid (post-substitution). In particular,
the solution (\ref{eq:ex1:sol}) precisely captures the
almost-\kw{Nat} property which allows an SMT solver to prove
(\ref{eq:ex1:2}) valid.
While the substituted VC contains existential
quantifiers, the SMT solver can handle the
resulting validity query easily as the
quantifiers appear \emph{negatively},
\ie in the antecedents of the
implication.
Consequently, they can be removed by a simple
variable renaming (``skolemization'') as the VC
$(\exists x. P(x)) \Rightarrow Q$
is equivalent to the VC
$P(z) \Rightarrow Q$,
where \kw{z} is a fresh variable name.

Thus, for local variables, \toolname synthesizes
the same intermediate types as the existentials-based
method of \cite{Knowles09,Kent16}, and, unlike Liquid Types,
is able to verify \kw{ex1} without requiring \emph{any}
predicate templates.
However, as we see next, \toolname generalizes
the existentials based approach to handle collections,
lambdas, and polymorphic, higher order functions.

\subsubsection*{Example 2: Collections} \hfill 
Next, let us see how \toolname precisely
synthesizes types for intermediate collections
without requiring user-defined qualifiers or
annotations.

\mypara{1. Templates}
%
%
From the output types of \kw{inc} and \kw{dec} we get
\begin{align*}
\kw{n}   & \ \dcolon\ \reftpv{\vv}{\tint}{\vv = \kw{x} - 1} \\
\kw{p}   & \ \dcolon\ \reftpv{\vv}{\tint}{\vv = \kw{x} + 1} \\
\intertext{The lists \kw{xs} and \kw{ys} are created by
invoking the nil and cons constructors (Fig~\ref{fig:lib}.)
In the two cases, we respectively instantiate the polymorphic
type variable \kw{a} (in the constructors' signatures from Figure~\ref{fig:lib})
with the (unknown) templates over fresh $\kvar$ variables
to get}
\kw{xs}  & \ \dcolon\ \tlist{\reftpv{\vv}{\tint}{\kva{\kvar_x}{\vv}}} \\
\kw{ys}  & \ \dcolon\ \tlist{\reftpv{\vv}{\tint}{\kva{\kvar_y}{\vv}}} \\
\intertext{as the output of \kw{last} has the same type its input list, we get:}
\kw{y}     & \ \dcolon\ \reftpv{\vv}{\tint}{\kva{\kvar_y}{\vv}} \\
\kw{inc y} & \ \dcolon\ \reftpv{\vv}{\tint}{\vv = \kw{y} + 1}
\end{align*}

\mypara{2. Constraints}
We get four implication constraints from \kw{ex2}.
\begin{alignat}{2}
& \csbind{\kw{x}}{0 \leq \kw{x}} \Rightarrow                                                          && \notag    \\
& \quad \wwedge\ \csbind{\kw{n}}{\kw{n} = \kw{x} - 1} \Rightarrow                                     && \notag    \\
& \quad \quad \quad    \csbind{\kw{p}}{\kw{p} = \kw{x} + 1} \Rightarrow                               && \notag    \\
& \quad \quad \quad    \quad \wedge\ \csbind{\vv}{\vv = \kw{n}}       \Rightarrow \kva{\kvar_x}{\vv} && \label{eq:ex2:1} \\
& \quad \quad \quad    \quad \wedge\ \csbind{\vv}{\vv = \kw{p}}       \Rightarrow \kva{\kvar_y}{\vv} && \label{eq:ex2:2} \\
& \quad \quad \quad    \quad \wedge\ \csbind{\vv}{\kva{\kvar_x}{\vv}} \Rightarrow \kva{\kvar_y}{\vv} && \label{eq:ex2:3} \\
& \quad \wedge\  \csbind{\kw{y}}{\kva{\kvar_y}{\kw{y}}}  \Rightarrow                                  && \notag    \\
& \quad \quad    \quad \csbind{\vv}{\vv = \kw{y} + 1}    \Rightarrow 0 \leq \vv                       && \label{eq:ex2:4}
\end{alignat}
As \kw{n} is passed into ``cons'' to get \kw{xs},
we get that $\reft{\vv}{\vv = \kw{n}}$ must be a subtype of
$\reft{\vv}{\kva{\kvar_x}{\vv}}$, yielding (\ref{eq:ex2:1}).
Similarly, as \kw{p} is ``cons''-ed to \kw{xs} to get \kw{ys},
we get that $\reft{\vv}{\vv = \kw{p}}$ and $\reft{\vv}{\kva{\kvar_x}{\vv}}$
must be subtypes of $\reft{\vv}{\kva{\kvar_y}{\vv}}$ yielding
(\ref{eq:ex2:2}) and (\ref{eq:ex2:3}) respectively.
Finally, the type of the return value \kw{y+1} must be a
subtype of \kw{Nat} yielding (\ref{eq:ex2:4}).
Each implication appears under a context
in which the variables in scope are bound to
satisfy their template's refinement.

\mypara{3. Solution}
Again, we compute the strongest refinements
as the (existentially quantified) disjunctions
of the hypotheses under which each $\kvar$ appears.
Thus, implication (\ref{eq:ex2:1}) yields:
\begin{alignat*}{3}
\kva{\kvar_x}{z}
  & \doteq\ && \exbind{\kw{x}}{0 \leq \kw{x}} \ \wedge  & \\
  &         && \quad \exbind{\kw{n}}{\kw{n} = \kw{x} - 1}\ \wedge &  \\
  &         && \quad \quad \exbind{\kw{p}}{\kw{p} = \kw{x} + 1}\ \wedge & \\
  &         && \quad \quad \quad \exbind{\vv}{\vv = n \ \wedge\ \vv = z} & \ \ldots \mbox{from}\ (\ref{eq:ex2:1}) \\
\intertext{which is essentially $0 \leq z + 1$, or almost-\kw{Nat}, and the disjunction
of (\ref{eq:ex2:2}, \ref{eq:ex2:3}) yields}
\kva{\kvar_y}{z}
  & \doteq\ && \exbind{\kw{x}}{0 \leq \kw{x}} \ \wedge  & \\
  &         && \quad \exbind{\kw{n}}{\kw{n} = \kw{x} - 1}\ \wedge &  \\
  &         && \quad \quad \exbind{\kw{p}}{\kw{p} = \kw{x} + 1}\ \wedge & \\
  &         && \quad \quad \quad \vee\ \exbind{\vv}{\vv = \kw{p} \ \wedge\ \vv = z}  & \ \ldots \mbox{from}\ (\ref{eq:ex2:2}) \\
  &         && \quad \quad \quad \vee\  \exbind{\vv}{\vv = \kw{n} \ \wedge\ \vv = z} & \ \ldots \mbox{from}\ (\ref{eq:ex2:3})
\end{alignat*}
which is also $0 \leq z + 1$. Substituting the strongest refinements
into (\ref{eq:ex2:4}) yields a valid formula, verifying \kw{ex2}.

\begin{figure}[t!]
\centering
\begin{minipage}[t]{0.45\textwidth}
\begin{codebox}
ex3 :: Nat -> Nat
ex3 = let fn = \a -> dec a
          fp = \b -> inc b
      in  fp . fn
\end{codebox}
\end{minipage}
\hspace{0.1in}
\begin{minipage}[t]{0.45\textwidth}
\begin{codebox}
ex4 :: List Nat -> List Nat
ex4 = let fn = \a -> dec a
          fp = \b -> inc b
      in map fp . map fn
\end{codebox}
\end{minipage}
\caption{Examples: (\bkw{ex3}) Higher-Order Composition,
                   (\bkw{ex4}) Higher-Order Iteration.}
\label{fig:ex3}
\label{fig:ex4}
\end{figure}

\subsubsection*{Example 3: Composition}
Next, we describe how \toolname synthesizes precise
types for inner lambda-terms like those bound to
\kw{fn} and \kw{fp}, and simultaneously determines
how the polymorphic type variables for the \kw{(.)}
combinator can be instantiated in order to verify \kw{ex3}

\mypara{1. Templates}
\toolname scales up to polymorphic higher-order operators
like ``compose'' \kw{(.)} by using the same approach as
for collections: create $\kvar$-variables for the unknown
instantiated refinements, and then find the strongest solution.
In \kw{ex3}, this process works by respectively instantiating
the \kw{a}, \kw{b} and \kw{c} in the signature for \kw{(.)}
(Fig~\ref{fig:lib}) with fresh templates
$\reft{\vv}{\kva{\kvar_a}{\vv}}$,
$\reft{\vv}{\kva{\kvar_b}{\vv}}$, and
$\reft{\vv}{\kva{\kvar_c}{\vv}}$  respectively.
Consequently, the arguments and output to \kw{.} at this instance
get the templates:
\begin{alignat}{2}
\kw{fn}      & \ \dcolon\ \reftpv{a}{\tint}{\kva{\kvar_a}{a}} & \ \rightarrow \reftpv{\vv}{\tint}{\kva{\kvar_b}{\vv}} \label{eq:template:fn}\\
\kw{fp}      & \ \dcolon\ \reftpv{b}{\tint}{\kva{\kvar_b}{b}} & \ \rightarrow \reftpv{\vv}{\tint}{\kva{\kvar_c}{\vv}} \label{eq:template:fp}\\
\kw{fp . fn} & \ \dcolon\ \reftpv{a}{\tint}{\kva{\kvar_a}{a}} & \ \rightarrow \reftpv{\vv}{\tint}{\kva{\kvar_c}{\vv}} \notag
\end{alignat}

\mypara{2. Constraints}
Next, the bodies of \kw{fn} and \kw{fp} must be subtypes
of the above templates. By decomposing function subtyping
into input- and output- subtyping, we get the following
implications. We omit the trivial constraints on the input
types; the above correspond to checking the output types
in an environment assuming the stronger (super-) input type:
\begin{alignat}{3}
             & \csbind{a}{\kva{\kvar_a}{a}}
             && \ \Rightarrow\ \csbind{\vv}{\vv = a - 1}
             & \ \Rightarrow\ \kva{\kvar_b}{\vv}
             & \label{eq:comp:a-b} \\
\wedge \quad & \csbind{b}{\kva{\kvar_b}{b}}
             && \ \Rightarrow\ \csbind{\vv}{\vv = b + 1}
             & \ \Rightarrow\ \kva{\kvar_c}{\vv}
             & \label{eq:comp:b-c} \\
\intertext{Finally, \kw{fp . fn} must be a subtype
of the (function) type ascribed to \kw{ex3} yielding
input and output constraints:}
\wedge \quad & \csbind{\vv}{0 \leq \vv}
             && \ \Rightarrow\ \kva{\kvar_a}{\vv}
             &
             & \label{eq:comp:nat-a} \\
\wedge \quad & \csbind{\vv}{\kva{\kvar_c}{\vv}}
             && \ \Rightarrow\ 0 \leq \vv
             &
             & \label{eq:comp:c-nat}
\end{alignat}

\mypara{3. Solution}
Finally \toolname computes the strongest refinements by
assigning $\kvar_a$ to its single hypothesis:
\begin{align*}
\kva{\kvar_a}{z} &\ \doteq\  \exbind{\vv}{0 \leq \vv \wedge \vv = z} \\
\intertext{which simplifies to ${0 \leq z}$, which plugged into (\ref{eq:comp:a-b}) gives:}
\kva{\kvar_b}{z} &\ \doteq\  \exbind{a}{0 \leq a \wedge \exbind{\vv}{\vv = a - 1 \wedge \vv = z}}\\
\intertext{which simplifies to ${0 \leq z + 1}$, which in (\ref{eq:comp:b-c}) yields}
\kva{\kvar_c}{z} &\ \doteq\ \exbind{b}{0 \leq b + 1 \wedge \exbind{\vv}{\vv = b + 1 \wedge \vv = z}}\\
\intertext{which is just $0 \leq z$, rendering implication (\ref{eq:comp:c-nat}) valid.}
\end{align*}

\subsubsection*{Example 4: Iteration}
%
Finally, \toolname uses types to scale
up to higher-order iterators over collections.
As in \kw{ex3} we generate fresh templates for
\kw{fn} (\ref{eq:template:fn}) and
\kw{fp} (\ref{eq:template:fp}).
When applied to \kw{map} the above return as output
the templates:
\begin{alignat*}{2}
\kw{map fn}
  & \ \dcolon\ \tlist{\reft{a}{\kva{\kvar_a}{a}}}
  && \rightarrow \tlist{ \reft{\vv}{\kva{\kvar_b}{\vv}}}  \\
\kw{map fp}
  & \ \dcolon\ \tlist{\reft{b}{\kva{\kvar_b}{b}}}
  && \rightarrow \tlist{\reft{\vv}{\kva{\kvar_c}{\vv}}} \\
\intertext{which are the templates from \kw{ex3}
\emph{lifted} to lists, yielding} 
\kw{map fp . map fn}
  & \ \dcolon\ \tlist{\reft{a}{\kva{\kvar_a}{a}}}
  && \rightarrow \tlist{\reft{\vv}{\kva{\kvar_c}{\vv}}}
\end{alignat*}

Consequently, the subtyping constraints for \kw{ex4} are
the same as for \kw{ex3} but lifted to lists.
Co-variant list subtyping reduces those into the
exact same set of implications as \kw{ex3}.
Thus, \toolname computes the same strongest
refinements as in \kw{ex3} and hence,
verifies \kw{ex4}.

\subsection{Avoiding Exponential Blowup}
\label{sec:overview:blowup}

The reader may be concerned that substituting
the strongest solution may cause the formulas to
expand, leading to an exponential blowup.
Recall that, in our examples, we ``simplified'' the
formulas before substituting to keep the exposition
short.
In general this is not possible as it requires
expensive quantifier elimination.
We show that the concern is justified in theory,
and that a direct substitution strategy leads to
a blowup even in practice.
Fortunately, we show that a simple optimization
that uses the scoping structure of bindings that
we have carefully preserved in the NNF constraints
allows us to avoid blowups in practice.

\mypara{Local Refinement Typing is \exptime-Hard}
The bad news is that local refinement typing in general,
and thus, constraint solving in particular are both
\exptime-hard.
This result follows from the result that
reachability (\ie safety) of non-recursive
first order boolean programs is
\exptime-complete~\cite{GodefroidYannakakis13}.
We can encode the reachability problem directly as checking
a given refinement type signature over programs using just
boolean valued data, and so local refinement typing is
\exptime-hard.
Thus, deciding constraint satisfaction
is also \exptime-hard, which means that we cannot
avoid exponential blowups in general.
Propositional validity queries are in \nptime so
\exptime-hardness means we must make an
exponential number of queries, or a polynomial
number of exponential size queries.

\mypara{Let-chains cause exponential blowup}
Indeed, the direct ``unfolding'' based approach
\cite{burstall-darlington-77,TamakiSato84} that
we have seen so far blows up due to a simple
and ubiquitous pattern: sequences of let-binders.
Consider the constraint generated by the program:
\begin{code}
  exp :: Nat -> Nat
  exp $\kw{x}_0$ = let $\kw{x}_1$ = id $\kw{x}_0$
               $\vdots$
               $\kw{x}_n$ = id $\kw{x}_{n-1}$
           in $\kw{x}_n$
\end{code}
The $i^{th}$ invocation of \kw{id :: a -> a}
creates a fresh $\kvar_i$, which is then
used as the template for $\kw{x}_i$.
Consequently, we get the NNF constraint:
\begin{equation}
\begin{array}{l@{\hskip 0.02in}l@{\hskip 0.02in}l@{\hskip 0.02in}l}
\csbind{\kw{x}_0}{0 \leq \kw{x}_0} & \multicolumn{3}{l}{\Rightarrow \csbind{\vv}{\vv = \kw{x}_0} \ \Rightarrow\  \kva{\kvar_1}{\vv}} \\
                                   & \ \ \wedge\ \csbind{\kw{x}_1}{\kva{\kvar_1}{\kw{x}_1}} & \multicolumn{2}{l}{\Rightarrow \csbind{\vv}{\vv = \kw{x}_1} \ \Rightarrow\  \kva{\kvar_2}{\vv}} \\
                                   &              & \ \ \vdots & \\
                                   &              & \ \ \wedge\ \csbind{\kw{x}_{n-1}}{\kva{\kvar_{n-1}}{\kw{x}_{n-1}}} & \Rightarrow \csbind{\vv}{\vv = \kw{x}_{n-1}} \ \Rightarrow\  \kva{\kvar_n}{\vv} \\
                                   &              & & \ \ \wedge\ \csbind{\kw{x}_n}{\kva{\kvar_n}{\kw{x}_n}}\ \Rightarrow \csbind{\vv}{\vv = \kw{x}_n} \ \Rightarrow\  0 \leq \vv
\end{array} \label{eq:cstr-blowup}
\end{equation}
%
Since each $\kvar_i$ has in its hypotheses,
all of $\kvar_1 \ldots \kvar_{i-1}$, the
strongest solution for $\kvar_i$ ends up
with $2^i$ copies of ${0 \leq x_0}$!
While this example is contrived, it
represents an idiomatic pattern: long
sequences of let-binders (\eg introduced
by ANF-conversion) with polymorphic
instantiation.
Due to the ubiquity of the pattern,
a direct unfolding based computation
of the strongest refinement fails for
all but the simplest programs.

\mypara{Sharing makes Solutions Compact}
This story has a happy ending.
Fortunately, our constraints
preserve the scoping structure
of binders.
Thus, the exponentially duplicated
$0 \leq x_0$ is ``in scope''
everywhere each $\kvar_i$
appears.
Consequently, we can \emph{omit}
it entirely from the strongest
solutions, collapsing each to
the compact solution (\ref{eq:compact-sol}),
$\kva{\kvar_i}{z} \doteq z = x_{i-1}$.
In \S~\ref{sec:constraints} we
formalize our algorithm for
generating such nested constraints,
and in \S~\ref{sec:algorithm} we
show how to use the above insight
to derive an optimized solving
algorithm, which yields compact,
shared solutions and hence,
dramatically improves
refinement type checking
in practice \S~\ref{sec:eval}.

\mypara{The Importance of Preserving Scope}
Note that the sharing observation and
optimization seem obvious at this juncture,
precisely because of our novel NNF formulation
which carefully preserves the scoping structure
of binders.
Previous approaches for
constraint based refinement
synthesis \cite{Knowles07,LiquidPLDI08,UnnoHORN15,Polikarpova16}
yield constraints that
discard the scoping
structure, yielding the
flat (\ie non-nested)
constraints
$$
\begin{array}{rl}
       & \csbind{\kw{x}_{0,0}}{0 \leq \kw{x}_{0,0}} \Rightarrow
         \csbind{\vv}{\vv = \kw{x}_{0,0}} \ \Rightarrow\  \kva{\kvar_1}{\vv} \\
\wedge & \csbind{\kw{x}_{1,0}}{0 \leq \kw{x}_{1,0}} \Rightarrow
         \csbind{\kw{x}_{1,1}}{\kva{\kvar_1}{\kw{x}_{1,1}}} \Rightarrow
         \csbind{\vv}{\vv = \kw{x}_{1,1}} \ \Rightarrow\  \kva{\kvar_2}{\vv} \\
\vdots & \\
\wedge & \csbind{\kw{x}_{n-1,0}}{0 \leq \kw{x}_{n-1,0}} \Rightarrow \ldots \Rightarrow
         \csbind{\kw{x}_{n-1,n-1}}{\kva{\kvar_{n-1}}{\kw{x}_{n-1,n-1}}} \Rightarrow
         \csbind{\vv}{\vv = \kw{x}_{n-1,n-1}} \ \Rightarrow\  \kva{\kvar_n}{\vv} \\
\wedge & \csbind{\kw{x}_{n,0}}{0 \leq \kw{x}_{n,0}} \Rightarrow \ldots \Rightarrow
         \csbind{\kw{x}_{n,n-1}}{\kva{\kvar_{n-1}}{\kw{x}_{n,n-1}}} \Rightarrow
         \csbind{\kw{x}_{n,n}}{\kva{\kvar_n}{\kw{x}_{n,n}}}\ \Rightarrow
         \csbind{\vv}{\vv = \kw{x}_{n,n}} \ \Rightarrow\  0 \leq \vv
\end{array}
$$
in which the sharing is not explicit as the
shared variables are alpha-renamed.
Absent sharing, unfolding yields a
solution for each $\kvar_i$ that is
exponential in $i$ without any syntactic
duplication of a common conjunct,
rendering the computation of precise
solutions infeasible.

\subsection{Relatively Complete Local Refinement Typing}

By \emph{fusing} together the constraints appearing
as hypotheses for each refinement variable $\kvar$,
our approach synthesizes the most precise
refinements expressible in the decidable
refinement logic (Lemma~\ref{thm:strongest}),
and hence, makes local refinement typing
\emph{relatively complete}.
That is, (assuming no recursion)
\emph{if there exist} type ascriptions
for intermediate binders and terms that
allow a program to typecheck, then \toolname
is guaranteed to find them.
The above examples \emph{do} in fact
typecheck with existing systems but
only if the user annotated local
variables with their signatures.
For example, \cite{fstar} \emph{can}
check \kw{ex2} (Figure~\ref{fig:ex2})
\emph{if} the programmer annotates
the type of the local \kw{xs} as an
almost \kw{Nat}.
Unfortunately, such superfluous
annotations impose significant
overhead in code size as well
as effort: the user must painstakingly
debug false alarms to find where
information was lost due to incompleteness.
In contrast, our completeness guarantee
means that once the user has annotated
top-level and recursive binders, \toolname
\emph{will not} return false alarms due
to missing templates, \ie will typecheck
the program if and only if it can be,
which crucially makes verification concise,
precise and decidable, enabling the results
in~\S~\ref{sec:eval}.

\mypara{Precise Refinements vs. Pre- and Post-Conditions}
%
%
As discussed in \S~\ref{sec:related},
our notion of \emph{most precise refinements}
can be viewed as the analog of the method of
strongest postconditions from Floyd-Hoare
style program logics.
\fstar introduces the notion of Dijkstra Monads
\cite{dijkstramonad} as a way to generalize
the dual notion of weakest preconditions (WP)
to the higher-order setting.
Hence, \fstar can use the notion of Dijkstra Monads to
verify the following variant of \kw{ex3} that uses
Floyd-Hoare style specification instead of refinement types:
\begin{fscode}
  let inc x = x + 1
  let dec x = x - 1
  let compose f g x = f (g x)

  let ex3 : i:int -> Pure int
    (requires (0 <= i)) (ensures (fun j -> 0 <= j))
    = compose inc dec
\end{fscode}
However, this approach has a key limitation
relative to \toolname in that it requires an
implementation or logically equivalent
\emph{strong} specification of the \kw{compose}
operator \kw{(.)}, instead of the weaker
polymorphic type signature required by \toolname.
The WP machinery crucially relies upon the
strong specification to makes the chaining
of functions \emph{explicit}.
When the chaining is obscured, \eg by the
use of \kw{map} in \kw{ex4}, the WP method
is insufficient, and so \fstar fails to
verify \kw{ex4}.
In contrast, \toolname uses only the \emph{implicit}
subtyping dependencies between refinement types.
Hence, the same compositional, type-directed
machinery that checks \kw{ex3} carries over
to successfully verify \kw{ex4}.

\section{Programs}\label{sec:lang}

\begin{figure}[t!]

\begin{tabular}{>{$}r<{$} >{$}r<{$} >{$}r<{$} >{$}l<{$} >{$}l<{$} }
\emphbf{Values}	     & \val   & ::=    & \vconst & \mbox{constants}\\
                     &        & \spmid & \elambda{\evar}{\utype}{\expr} & \mbox{function}\\[0.05in]

\emphbf{Terms}       & \expr  & ::=    & \val                                     & \mbox{values} \\
                     &        & \spmid & \evar, \ssymbol{y}, \ssymbol{z}, \ldots  & \mbox{variables} \\
                     &        & \spmid & \eletin{\evar}{\expr}{\expr'}            & \mbox{let-bind} \\
                     &        & \spmid & \app{\expr}{\evar}                       & \mbox{application} \\
                     &        & \spmid & \tlambda{\tvar}{\expr}                   & \mbox{type-abstraction} \\
                     &        & \spmid & \tapp{\expr}{\utype}                     & \mbox{type instantiation} \\ [0.05in]

\emphbf{Refinements} & \refi  & ::= & \dots \text{varies} \dots & \\ [0.05in]

\emphbf{Basic Types} & \base  & ::= & \tint \spmid \tbool \spmid \dots & \\ [0.05in]

\emphbf{Types}       & \type  & ::=    & \tvar                      & \mbox{variable} \\
                     &        & \spmid & \reftpv{x}{\base}{\refi}   & \mbox{base}  \\
                     &        & \spmid & \trfun{x}{\type}{\type}   & \mbox{function}      \\
                     &        & \spmid & \tpoly{\tvar}{\type}       & \mbox{scheme}        \\
                     [0.05in]

\emphbf{Unrefined Types} & \utype & ::=    & \tvar                     & \mbox{variable}\\
                      &        & \spmid & \base                     & \mbox{base}\\
                      &        & \spmid & \trfun{x}{\utype}{\utype} & \mbox{function} \\
                      &        & \spmid & \tpoly{\tvar}{\utype}     & \mbox{scheme}        \\
\end{tabular}
\caption{{Syntax of \lang}}
\label{fig:syntax}
\end{figure}

We start by formalizing the syntax and
operational semantics of a core source
language \lang ~\cite{Knowles10, LiquidICFP14}.

\subsection{Syntax}

\mypara{Terms}
Figure~\ref{fig:syntax} summarizes the syntax of
the terms of \lang.
The \emph{values} of \lang comprise primitive
constants $\vconst$ and functions
$\elambda{\evar}{\utype}{\expr}$.
The \emph{terms} of \lang include values and
additionally,
variables \ssymbol{x}, \ssymbol{y}, \ssymbol{z} \ldots,
let-binders \eletin{\evar}{\expr}{\expr},
applications \app{\expr}{\evar} (in ANF,
to simplify the application rule \cite{LiquidPLDI08}).
Polymorphism is accounted for via type
abstraction \tlambda{\tvar}{\expr} and
instantiation \tapp{\expr}{\utype}.
%
We assume that the source program has been annotated with
\emph{unrefined} types at $\lambda$-abstractions and type
abstraction and instantiations, either by the programmer
or via classical (unrefined) type inference.

\mypara{Types}
The types of \lang include unrefined types,
written $\utype$, and refined types, written
$\type$.
The unrefined types include basic types like
$\tint$, $\tbool$, type variables $\tvar$,
functions $\trfun{\evar}{\utype}{\utype'}$ where
the input is bound to the name \evar, and type
schemes $\tpoly{\tvar}{\utype}$.
We refine basic types with
\emph{refinement predicates}
$\refi$ to get refined types.
The formulas $\refi$ are
drawn from a decidable
refinement logic, \eg \decidablelogic: the
quantifier free theory of linear arithmetic
and uninterpreted functions.

\mypara{Notation}
We write \reft{\vv}{\refi}
when the base \base is clear from the context.
We write \base to abbreviate
\reftpv{\vv}{\base}{\ttrue},
and write
${\type \rightarrow \type'}$ to
abbreviate \trfun{\evar}{\type}{\type'}
if $\evar$ does not appear in $\type'$.

\mypara{Typing Constants}
We assume that each constant \vconst
is equipped with a type \constty{\vconst}
that characterizes its semantics~\cite{Knowles10,LiquidICFP14}.
For example, literals are assigned their corresponding
singleton type and operators have types representing
their pre- and post-conditions:
\begin{align*}
  \constty{\kw{7}} & \ \doteq\ \reftpv{\vv}{\tint}{\vv = \kw{7}} \\
  \constty{\kw{+}} & \ \doteq\ \trfun{x}{\tint}{\trfun{y}{\tint}{\reft{\vv}{\vv = x + y}}}\\
  \constty{\kw{/}} & \ \doteq\ \tint \rightarrow \reftpv{\vv}{\tint}{\vv \not = \kw{0}} \rightarrow \tint\\
  \constty{\kw{assert}} & \ \doteq\ \reftpv{\vv}{\tbool}{\vv} \rightarrow \tunit
\end{align*}

\subsection{Semantics}
\lang has a standard small-step, contextual
\emph{call-by-value} semantics; we write
$\steps{\expr}{\expr'}$ to denote that
term $\expr$ steps to $\expr'$.
We write \stepsn{\expr}{\expr'}{j} if there
exists $e_1,\ldots,e_j$ such that
$\expr \equiv \expr_1$,
$\expr' \equiv \expr_j$,
and for all $1 \leq i < j$,
we have $\steps{\expr_i}{\expr_{i+1}}$.
We write \stepsmany{\expr}{\expr'} if
for some (finite) $j$ we have
\stepsn{\expr}{\expr'}{j}.

\mypara{Constants}
We assume that when values are applied
to a primitive constant, the expression
is reduced to the output of the primitive
constant operation in a single step.
For example, consider \kw{+}, the
primitive \tint addition operator.
We assume that
$\primapp{\kw{+}}{n}$ equals $\kw{+}_n$
where for all $m$,
\primapp{\kw{+}_n}{m} equals the
integer sum of $n$ and $m$.

\mypara{Safety}
We say a term \expr \emph{is stuck} if
there does not exist any $\expr'$ such
that $\steps{\expr}{\expr'}$.
We say that a term \expr \emph{is safe},
if whenever $\stepsmany{\expr}{\expr'}$
then either $\expr'$ is a value
or $\expr'$ is \emph{not} stuck.
We write \issafe{\expr} to denote that
\expr is safe.
Informally, we assume that the primitive operations
(\eg division) get stuck on values that do not satisfy
their input refinements (\eg when the divisor is \kw{0}).
Thus, $\expr$ is safe when every (primitive) operation
is invoked with values on which it is defined (\eg there
are no divisions by \kw{0}.)

\section{Constraints}\label{sec:constraints}

\begin{figure}[t!]
\judgementHead{Constraint Syntax}{c}

\begin{tabular}{>{$}r<{$} >{$}r<{$} >{$}r<{$} >{$}l<{$}}
\emphbf{Predicates}  & \pred  & ::=    & \refi \\
                     &        & \spmid & \kvapp{\kvar}{y} \\
                     &        & \spmid & \pred_1 \wedge \pred_2 \\[0.05in]

\emphbf{Constraints} & \cstr  & ::=    & \pred \\
                     &        & \spmid & \csand{\cstr_1}{\cstr_2} \\
                     &        & \spmid & \csimp{x}{\base}{\pred}{\cstr}\\[0.05in]

\emphbf{Environments}& \tcenv & ::=    & \emptyset \spmid \tcenvext{x}{\type} \\[0.05in]

\emphbf{Assignment}   & \soln  & ::=    & \Kvar \rightarrow \Refi \\[0.05in]

\emphbf{Assumption}   & \assm  & ::=    & \soln \spmid \assmext{\assm}{x}{\base}{\pred} \\
\end{tabular}

\vspace{0.1in}

\judgementHead{Well-formedness}{\wf{\tcenv}{\cstr}}
$$\begin{array}{lll}
\inferrule[\rulename{W-Ref}]
  {\hastype{\tcenv}{\refi}{\tbool}}
  {\wf{\tcenv}{\refi}}
&
\inferrule[\rulename{W-And}]
  {\wf{\tcenv}{\cstr_1} \quad \wf{\tcenv}{\cstr_2}}
  {\wf{\tcenv}{\csand{\cstr_1}{\cstr_2}}}
&
\inferrule[\rulename{W-Imp}]
  {\wf{\tcenvext{x}{\base}}{\cstr} \quad
   \wf{\tcenvext{x}{\base}}{\refi}
  }
  {\wf{\tcenv}{\csimp{x}{\base}{\refi}{\cstr}}}
\end{array}$$

\vspace{0.1in}

\judgementHead{Satisfaction}{\satisfies{\assm}{\cstr}}
$$\begin{array}{lll}
\inferrule[\rulename{Sat-Base}]
  { \wf{\emptyset}{\apply{\soln}{\cstr}} \quad  \smtvalid{\apply{\soln}{\cstr}}}
  {\satisfies{\soln}{\cstr}}
&
\inferrule[\rulename{Sat-Ext}]
  {\satisfies{\assm}{\csimp{x}{\base}{\pred}{\cstr}}}
  {\satisfies{\assmext{\assm}{x}{\base}{\pred}}{\cstr}}
&
\inferrule[\rulename{Sat-Many}]
  {\satisfies{\assm}{\cstr}\ \mbox{for each}\ \cstr \in \cstrs}
  {\satisfies{\assm}{\cstrs}}
\end{array}$$
\caption{{Constraints: Syntax and Semantics}}
\label{fig:constraints:syntax}
\label{fig:constraints:wf}
\label{fig:satisfaction}
\end{figure}

Local refinement typing proceeds in two steps.
First, we use \lang terms to generate a system
of constraints (\S~\ref{sec:constraint:syntax})
whose satisfiability (\S~\ref{sec:constraint:semantics})
implies that the term is safe (\S~\ref{sec:constraint:generation}).
Second, we solve the constraints to find a satisfying assignment (\S~\ref{sec:algorithm}).

\subsection{Syntax}
\label{sec:constraint:syntax}

\mypara{Refinement Variables}
Figure~\ref{fig:constraints:syntax} describes
the syntax of constraints.
A \emph{refinement variable} represents
an unknown $n$-ary relation, \ie is an
unknown refinement over the free variables
$z_1,\ldots,z_n$ (abbreviated to $\params{z}$).
We refer to $\params{\tb{z}{\utype}}$ as
the \emph{parameters} of a variable $\kvar$,
written $\args{\kvar}$, and say that $n$ is
the \emph{arity} of $\kvar$.
We write $\Kvar$ for the set of all
refinement variables.

\mypara{Predicates}
A \emph{predicate} is a refinement
$\refi$ from a decidable logic, or
a \emph{refinement variable}
application $\kva{\kvar}{x_1,\ldots,x_n}$,
or a conjunction of two sub-predicates.
We assume each $\kvar$ is applied to
the same number of variables as its arity.

\mypara{Constraints}
A constraint is a ``tree'' where each
``leaf'' is a \emph{goal} and each
``internal'' node either
(1) universally \emph{quantifies} some variable $x$
    of a basic type $\base$, subjecting it to
    satisfy a \emph{hypothesis} $\pred$, \emph{or}
(2) \emph{conjoins} two sub-constraints.

\mypara{NNF Horn Clauses}
Our constraints are Horn Clauses in
Negation Normal Form~\cite{hornsurvey}.
NNF constraints are a generalization
of the ``flat'' Constraint Horn Clause (CHC)
formulation used in Liquid Typing~\cite{LiquidPLDI08}.
Intuitively, each root-to-leaf path in the tree
is a CHC that requires that the leaf's goal be
implied by the internal nodes' hypotheses.
The nesting structure of NNF permits
the scope-based optimization that makes \toolname
practical~(\S~\ref{sec:algo:scope}).

\subsection{Semantics}
\label{sec:constraint:semantics}

We describe the semantics of constraints by formalizing
the notion of constraint satisfaction. Intuitively, a
constraint is satisfiable if there is an interpretation
or \emph{assignment} of the $\kvar$-variables to concrete
refinements $\refi$ such that the resulting logical
formula is \emph{well-formed} and \emph{valid}.

\mypara{Well-formedness}
An \emph{environment} $\tcenv$ is a set of type
bindings $\tb{x}{\utype}$.
We write $\wf{\tcenv}{\cstr}$ to denote that
$\cstr$ is \emph{well-formed} under environment
$\tcenv$.
Figure~\ref{fig:constraints:wf} summarizes the
rules for establishing well-formedness.
We write $\hastype{\tcenv}{\refi}{\tbool}$ if
$\refi$ can be typed as a \tbool using the
standard (non-refined) rules, and use it to
check individual refinements (\rulename{W-Ref}).
A conjunction of sub-constraints is well-formed if
each conjunct is well-formed (\rulename{W-And}).
Finally, each implication is well formed if
the hypothesis is well-formed and the consequent
is well-formed under the extended environment
(\rulename{W-Imp}).
(Note that a constraint containing a $\kvar$-application
is \emph{not} well-formed.)

\mypara{Assignments}
An \emph{assignment} $\soln$ is a map from the
set of refinement variables $\Kvar$ to the set of
refinements $\Refi$.
An assignment is \emph{partial} if its range is predicates
(containing $\kvar$-variables, not just refinements).
The function $\apply{\soln}{\cdot}$
substitutes $\kvar$-variables in
predicates $\pred$,
constraints $\cstr$,
and sets of constraints $\Cstr$
with their $\soln$-assignments:
$$\begin{array}{lcl}
\apply{\soln}{\kvapp{\kvar}{y}}
  & \doteq & \SUBST{\soln(\kvar)}{\params{x}}{\params{y}}
    \quad \mbox{where $\params{x} = \args{\kvar}$}        \\
\apply{\soln}{\refi}
  & \doteq & \refi                                         \\
\apply{\soln}{\pred \wedge \pred'}
  & \doteq & \apply{\soln}{\pred} \wedge \apply{\soln}{\pred'} \\
\apply{\soln}{\csimp{x}{\base}{\pred}{c}}
  & \doteq & \csimp{x}{\base}{\apply{\soln}{\pred}}{\apply{\soln}{\cstr}} \\
\apply{\soln}{\csand{\cstr_1}{\cstr_2}}
  & \doteq & \csand{ \apply{\soln}{\cstr_1}   }{ \apply{\soln}{\cstr_2}  } \\
\apply{\soln}{\Cstr}
  & \doteq & \{ \apply{\soln}{\cstr} \spmid \cstr \in \Cstr \}
\end{array}$$


\mypara{Satisfaction}
A \emph{verification condition} (VC) is a constraint
$\cstr$ that has no $\kvar$ applications; \ie
$\kvars{\cstr} = \emptyset$.
We say that an assignment $\soln$ \emph{satisfies}
constraint $\cstr$, written $\satisfies{\soln}{\cstr}$,
if $\apply{\soln}{\cstr}$ is a VC that is:
(a)~\emph{well-formed}, and
(b)~\emph{logically valid} \cite{NelsonOppen}.
Figure~\ref{fig:satisfaction} formalizes the notion
of satisfaction and lifts it to sets of constraints.
We say that $\cstr$ \emph{is satisfiable} if
\emph{some} assignment $\soln$ satisfies $\cstr$.
The following lemmas follow from the definition
of satisfaction and validity:

\begin{lemma}[\textbf{Weaken}] \label{lem:weakening} 
If
  $\satisfies{\assm}{\pred}$
then
  $\satisfies{\assm}{\pred~\vee~\pred'}$.
\end{lemma}

\begin{lemma}[\textbf{Strengthen}] \label{lem:strengthening} 
If
  $\satisfies{\soln}{\csimp{x}{\base}{\pred}{\cstr}}$ and
  $\satisfies{\soln}{\csimp{x}{\base}{\pred'}{\pred}}$
then
  $\satisfies{\soln}{\csimp{x}{\base}{\pred'}{\cstr}}$.
\end{lemma}

\mypara{Composing Assignments}
For two assignments $\soln$ and $\soln'$ we write
$\comp{\soln}{\soln'}$ to denote the assignment:
$$\apply{(\comp{\soln}{\soln'})}{\kvar} \ \doteq\ \apply{\soln}{\apply{\soln'}{\kvar}}$$
The following Lemmas, proved by induction
on the structure of $\cstr$, characterize
the relationship between assignment
composition and satisfaction.

\begin{lemma}[\textbf{Composition-Preserves-Sat}] \label{lem:sat-composition} 
If $\satisfies{\soln}{\cstr}$ then $\satisfies{\comp{\soln'}{\soln}}{\cstr}$.
\end{lemma}

\begin{lemma}[\textbf{Composition-Validity}] \label{lem:valid-composition} 
$\satisfies{\comp{\soln}{\soln'}}{\cstr}$ if and only if $\satisfies{\soln}{\apply{\soln'}{\cstr}}$. 
\end{lemma}

\subsection{Dependencies and Decomposition}
\label{sec:decomposition}

Next, we describe various properties of constraints
that we will exploit to determine satisfiability.

\mypara{Flat Constraints}
A \emph{flat} or non-nested constraint $\cstr$ is of the form:
$$\csimp{x_1}{\base_1}{\pred_1}{\ldots~\Rightarrow~\csimp{x_n}{\base_n}{\pred_n}{\pred}}\label{eq:flat}$$
By induction on the structure of $\cstr$ we can
show that the procedure $\flatten{\cstr}$, shown
in Figure~\ref{fig:flatten}, returns a set of flat
constraints that 
are satisfiable exactly when $\cstr$ is satisfiable.
\begin{lemma}[\textbf{Flattening}] \label{lem:flatten}
$\satisfies{\soln}{\cstr}$
if and only if
$\satisfies{\soln}{\flatten{\cstr}}$.
\end{lemma}


\mypara{Heads and Bodies}
Let
$\cstr\equiv\csimp{x_1}{\base_1}{\pred_1}{\ldots\Rightarrow\csimp{x_n}{\base_n}{\pred_n}{\pred}}$
be a flat constraint.
We write $\head{\cstr}$ for the predicate $\pred$
that is the right-most ``consequent'' of $\cstr$.
We write $\body{\cstr}$ for the set of predicates
$\pred_1,\ldots,\pred_n$ that are the ``antecedents''
of $\cstr$.

\mypara{Dependencies}
Let $\kvars{\cstr}$ denote the
the set of $\kvar$-variables
in a constraint $\cstr$.
We can define the following sets of \emph{dependencies}: 
\begin{align*}
\dep{\soln} & \doteq \{ (\kvar, \kvar') \spmid \kvar \in \apply{\soln}{\kvar'} \}
  \tag{of a \emph{solution} $\soln$} \\
\dep{\cstr} & \doteq \{ (\kvar, \kvar') \spmid \kvar \in \body{\cstr}, \kvar' \in \head{\cstr} \}
  \tag{of a \emph{flat constraint} $\cstr$} \\
\dep{\cstr} & \doteq \cup_{\cstr' \in \flatten{\cstr}} \dep{\cstr'}
  \tag{of an \emph{NNF constraint} $\cstr$} \\
\deps{\Kvarc}{\cstr} & \doteq \dep{\cstr} \setminus (\Kvar \times \Kvarc \cup \Kvarc \times \Kvar)
  \tag{of a constraint \emph{excluding} $\Kvarc$}
\end{align*}
We write $\depstar{\cstr}$ (resp. $\depsstar{\Kvarc}{\cstr}$)
for the reflexive and transitive closure of
$\dep{\cstr}$ (resp. $\deps{\Kvarc}{\cstr}$).
Intuitively, the dependencies of a constraint comprise the
set of pairs $(\kvar, \kvar')$ where $\kvar$ appears in a
body (hypothesis) for a clause where $\kvar'$ is in the head (goal).

\mypara{Cuts and Cycles}
A set of variables $\Kvarc$ \emph{cuts} $\cstr$ if
the relation $\depsstar{\Kvarc}{\cstr}$ is \emph{acyclic}.
A set of variables $\Kvar'$ is \emph{acyclic} in $\cstr$
if $\Kvar - \Kvar'$ cuts $\cstr$.
Intuitively, a set $\Kvar'$ is acyclic in $\cstr$ if
the closure of dependencies \emph{restricted to} $\Kvar'$
(\ie excluding variables not in $\Kvar'$) has no cycles.
A single variable $\kvar$ is \emph{acyclic} in $\cstr$
if $\{ \kvar \}$ is acyclic in $\cstr$.
Intuitively, a single variable $\kvar$ is acyclic in $\cstr$
if there is no clause in $\cstr$ where $\kvar$ appears in both
the head and the body.
A constraint is \emph{acyclic} if the set of
\emph{all} $\Kvar$ is acyclic in $\cstr$,
\ie if $\depstar{\cstr}$ is acyclic.
In the sequel, for clarity of exposition we assume the
invariant that whenever relevant, $\kvar$ is acyclic in
$\cstr$. (Of course, our algorithm maintains this property,
as described in~\S~\ref{sec:algo:cyclic}.)

\mypara{Definitions and Uses}
The \emph{definitions} (resp. \emph{uses}) of $\kvar$ in $\cstr$,
written $\defns{\cstr}{\kvar}$ (resp. $\uses{\cstr}{\kvar}$),
are the subset of $\flatten{\cstr}$ such that the head
contains (resp. does not contain) $\kvar$:
\begin{align}
\defns{\cstr}{\kvar}
   & \doteq \ \{ \cstr' \mid \cstr' \in \flatten{\cstr},\ \kvar \in \head{\cstr'} \}
   \tag{definitions} \\
\uses{\cstr}{\kvar}
   & \doteq \ \{ \cstr' \mid \cstr' \in \flatten{\cstr},\ \kvar \not \in \head{\cstr'} \}
   \tag{uses}
\end{align}
By induction on the structure of $\cstr$ we can check that:
\begin{lemma}[\textbf{Flatten-Definitions}]\label{lem:flat-defn}
$\csimps{x}{\pred}{\defns{\cstr}{\kvar}} \equiv \defns{(\csimps{x}{\pred}{\cstr})}{\kvar}$
\end{lemma}

Since the definitions and uses of a $\kvar$ partition
$\flatten{\cstr}$, we can show that $\soln$ satisfies
$\cstr$ if it satisfies the definitions and uses of
$\kvar$ in $\cstr$:
\begin{lemma}[\textbf{Partition}]\label{lem:partition} 
  If $\kvar$ is acyclic in $\cstr$ then
  $\satisfies{\soln}{\cstr}$
  if and only if
  $\satisfies{\soln}{\defns{\cstr}{\kvar}}$ and
  $\satisfies{\soln}{\uses{\cstr}{\kvar}}$.
\end{lemma}

\subsection{Generation}
\label{sec:constraint:generation}

Next, we show how to map a term $\expr$
into a constraint $\cstr$ whose satisfiability
implies the safety of $\expr$.

\mypara{Shapes ($\shapesym$)}
Procedure $\shapesym$, elided for brevity, takes as input a refined
type $\type$ and returns as output the non-refined
version obtained by erasing the refinements
from $\type$.

\begin{figure}[t!]
$$\begin{array}{lcl}
\toprule
\freshsym & : & (\tcenv \times \Type) \rightarrow \liq{\Type} \\
\midrule
\fresh{[\tb{y_1}{\utype_1},\ldots]}{\base} & \doteq & \reftpv{\vv}{\base}{\kva{\kvar}{y_1,\ldots,\vv}} \\
\quad \mbox{where}       &   & \\
\quad \quad \kvar        & = & \mbox{fresh $\Kvar$ variable} \\
\quad \quad \mbox{s.t.}\ \args{\kvar} & = & [\tb{z_1}{\utype_1},\ldots,\tb{z}{\base}] \\[0.05in]

\fresh{\tcenv}{\trfun{x}{\utype}{\utype'}} & \doteq & \trfun{x}{\type}{\type'} \\
\quad \mbox{where}       &   & \\
\quad \quad \type        & = & \fresh{\tcenv}{\utype} \\
\quad \quad \type'       & = & \fresh{\tcenvext{x}{\utype}}{\utype'} \\[0.05in]

\fresh{\tcenv} {\tvar}   & \doteq & \tvar  \\[0.05in]
\fresh{\tcenv}{\tpoly{\tvar}{\utype}}
                         & \doteq & \tpoly{\tvar}{\fresh{\tcenv}{\utype}} \\[0.1in]
\bottomrule
\end{array}$$
\caption{Fresh Templates: $\liq{\Type}$ denotes \emph{templates},
\ie types with $\kvar$-variables representing unknown refinements.}
\label{fig:fresh}
\end{figure}

\begin{figure}[t!]
$$\begin{array}{lcl}
\toprule
\subsym & : & (\Type \times \Type) \rightarrow \Cstr \\
\midrule
{\sub{\reftpv{x}{\base}{\pred}}{\reftpv{y}{\base}{\predb}}}
  & \doteq & \csimp{x}{\base}{\pred}{\SUBST{\predb}{y}{x}} \\[0.05in]
{\sub{\trfun{x}{\typeb}{\typeb'}}{\trfun{y}{\type}{\type'}}}
  & \doteq & \csand{\cstr}{(\cswith{y}{\type}{\cstr'})} \\
\quad \mbox{where}  &   & \\
\quad \quad \cstr   & = & \sub{\type}{\typeb} \\
\quad \quad \cstr'  & = & \sub{\SUBST{\typeb'}{x}{y}}{\type'} \\[0.05in]
{\sub{\tvar}{\tvar}}
  & \doteq & \ttrue \\[0.05in]
{\sub{\tpoly{\tvar}{\type}}{\tpoly{\tvar}{\type'}}}
  & \doteq & \sub{\type}{\type'} \\[0.05in]
\bottomrule
\end{array}$$
\caption{Subtyping Constraints}
\label{fig:sub}
\end{figure}

\mypara{Templates ($\freshsym$)}
Procedure $\fresh{\tcenv}{\utype}$ (Figure~\ref{fig:fresh})
takes as input
an environment and a non-refined type $\utype$, and
returns a template whose $\shapesym$ is $\utype$.
Recall that each refinement variable
denotes an $n$-ary relation. For example,
a refinement variable that appears in the
\emph{output} type of a function can refer
to the \emph{input} parameter, \ie can be
a \emph{binary} relation between the input
and the output value.
We track the allowed parameters by using the
environment $\tcenv$.
In the base case, $\freshsym$ generates
a fresh $\kvar$ whose parameters correspond
to the binders of $\tcenv$.
In the function type case, $\freshsym$
recursively generates templates for the
input and output refinements, extending
the environment for the output with the
input binder.

\mypara{Subtyping ($\subsym$)}
Procedure $\sub{\type}{\type'}$ (Figure~\ref{fig:sub})
returns the constraint $\cstr$ that must be satisfied
for $\type$ to be a subtype of $\type'$, which
intuitively means, the set of values denoted by
$\type$ to be subsumed by those denoted by $\type'$.
(See \cite{Knowles10, LiquidICFP14} for details.)
In the base case, the sub-type's predicate must
imply the super-type's predicate.
In the function case, we conjoin the
contra-variant input constraint and
the co-variant output constraint.
For the latter, we additionally add
a hypothesis assuming the stronger
input type, using a \emph{generalized
implication} that drops binders with
non-basic types as they are not
supported by the first-order
refinement logic:
$$\begin{array}{lcl}
\cswith{x}{\reftpv{y}{\base}{\pred}}{\cstr}
  & \doteq & \csimp{x}{\base}{\SUBST{\pred}{y}{x}}{\cstr} \\
\cswith{x}{\type}{\cstr}
  & \doteq & \cstr \\
\end{array}$$

\mypara{Generation ($\conssym$)}
Finally, procedure $\cons{\tcenv}{\expr}$
(Figure~\ref{fig:consgen}) takes as input
an environment $\tcenv$ and term $\expr$
and returns as output a pair $(\cstr, \type)$
where $\expr$ typechecks if $\cstr$ holds, and
has the (refinement) template $\type$.
The procedure computes the template
and constraint by a syntax-directed traversal
of the input term.
The key idea is two-fold.
First: generate a fresh template for each
position where the (refinement) types
cannot be directly synthesized -- namely,
function inputs, let-binders, and polymorphic
instantiation sites.
Second: generate subtyping constraints at each
position where values flow from one position into
another.

\begin{itemize}

\item \emphbf{Constants and Variables} yield the
trivial constraint $\ttrue$, and templates
corresponding to their primitive types, or
environment bindings respectively.
(The helper $\singtysym$ strengthens the
refinement to state the value equals the
given variable, dubbed ``selfification''
by \cite{Ou2004}, enabling path-sensitivity
~\cite{LiquidPLDI08,Knowles10}.)

\item \emphbf{Let-binders} ($\eletin{\evar}{\expr_1}{\expr_2})$
yield the conjunction of three sub-constraints.
First, we get a constraint $\cstr_1$ and template
$\type_1$ from $\expr_1$.
Second, we get a constraint $\cstr_2$ and template
$\type_2$ from $\expr_2$ checked under the environment
extended with $\type_1$ for $\evar$.
Third, to hide \evar which may appear in $\type_2$,
we generate a fresh template $\liq{\type}$ with the
shape of $\type_2$ and use $\subsym$ to constrain
$\type_2$ to be a subtype of $\liq{\type}$.
The well-formedness requirement ensures $\liq{\type}$,
which is returned as the expression's template, is
assigned a well-scoped refinement type.
\item \emphbf{Functions} ($\elambda{\evar}{\utype}{\expr}$)
yield a template where the input is assigned a
fresh template of shape $\utype$, and the output
is the template for \expr.

\item \emphbf{Applications} ($\app{\expr}{y}$) yield
the output template with the formal $x$ substituted
for the actual $y$. (ANF ensures the substitution does
not introduce arbitrary expressions into the refinements.)
Furthermore, the input value $y$ is constrained
to be a subtype of the function's input type.

\item \emphbf{Type-Instantiation} ($\tyapp{\expr}{\utype}$)
is handled by generating a fresh template
whose shape is that of the polymorphic instance
$\utype$, and substituting that in place of the
type variable $\tvar$ in the type (scheme) obtained
for $\expr$.

\end{itemize}

\mypara{Soundness of Constraint Generation}
We can prove by induction on the structure of $\expr$
that the NNF constraints generated by $\conssym$
capture the refinement type checking obligations
of~\cite{Knowles10,LiquidICFP14}.
Let $\conssym(\expr)$\ denote the constraint returned
by $\cons{\emptyset}{\expr}$.
We prove that if $\conssym(\expr)$ is satisfiable, then
${\hastype{\emptyset}{\expr}{\type}}$ for some refinement
type $\type$. This combined with the soundness of refinement
typing (Theorem~1,~\cite{LiquidICFP14}) yields:
\begin{theorem}\label{thm:consgen}
If $\conssym(\expr)$\ is satisfiable then $\expr$ is safe.
\end{theorem}

\begin{figure}[t!]
$$\begin{array}{lcl}
\toprule
\conssym & : &  (\tcenv \times \Expr) \rightarrow (\Cstr \times \Type) \\
\midrule
\cons{\tcenv}{\vconst}
  & \doteq & (\ttrue,\ \constty{\vconst})  \\[0.05in]
\cons{\tcenv}{\evar}
  & \doteq & ( \ttrue ,\ \singty{\tcenv}{\evar}) \\[0.05in]
\cons{\tcenv}{\eletin{\evar}{\expr_1}{\expr_2}}
                          & \doteq & ( \csand{\cstr}{\sub{\type_2}{\liq{\type}}} ,\ \liq{\type} )\\
\quad \mbox{where}        &        & \\
\quad \quad \cstr         &      = & \csand{\cstr_1}{(\cswith{x}{\type_1}{\cstr_2})}\\
\quad \quad \liq{\type}   &      = & \fresh{\emptyset}{\shape{\type_2}} \\
\quad \quad (\cstr_1, \type_1) & = & \cons{\tcenv}{e_1} \\
\quad \quad (\cstr_2, \type_2) & = & \cons{\tcenvext{x}{\type_1}}{e_2} \\[0.05in]

\cons{\tcenv}{\elambda{\evar}{\utype}{\expr}}
                          & \doteq & (\cstr, \trfun{\evar}{\liq{\type}}{\type}) \\
\quad \mbox{where}        &        & \\
\quad \quad (\cstr, \type)&      = & \cons{\tcenvext{\evar}{\liq{\type}}}{\expr} \\
\quad \quad \liq{\type}   &      = & \fresh{\emptyset}{\utype} \\[0.05in]

\cons{\tcenv}{\app{\expr}{y}}
                          & \doteq & (\csand{\cstr}{\cstr_y}, \SUBST{\type'}{x}{y}) \\
\quad \mbox{where}        &        & \\
\quad \quad \cstr_y       &      = & \sub{\singty{\tcenv}{y}}{\type}\\
\quad \quad (\cstr,
             \trfun{x}
                   {\type}
                   {\type'})   & = & \cons{\tcenv}{\expr} \\[0.05in]

\cons{\tcenv}{\tlambda{\tvar}{\expr}}
                          & \doteq & (\cstr, \tpoly{\tvar}{\type}) \\
\quad \mbox{where}        &        & \\
\quad \quad (\cstr,\type) &      = & \cons{\tcenv;\tvar}{\expr}    \\[0.05in]

\cons{\tcenv}{\tyapp{\expr}{\utype}}
                          & \doteq & (\cstr, \SUBST{\type}{\tvar}{\liq{\type}}) \\
\quad \mbox{where}        &        & \\
\quad \quad (\cstr, \tpoly{\tvar}{\type})
                          &      = & \cons{\tcenv}{\expr}    \\[0.05in]
\quad \quad \liq{\type}   &      = & \fresh{\emptyset}{\utype} \\[0.1in]

\toprule
\singtysym & : & (\tcenv \times \Evar) \rightarrow \Type \\
\midrule
\singty{\tcenv}{\evar} & & \\
\quad \spmid \type = \reftpv{x}{\base}{\pred} & \doteq & \reftpv{\vv}{\base}{\SUBST{p}{x}{\vv} \wedge \vv = x} \\
\quad \spmid \mbox{otherwise}                 & \doteq & \type \\
\quad \mbox{where}\ \type                     &      = & \tcenv(\evar) \\[0.1in]
\bottomrule
\end{array}$$
\caption{Constraint Generation
%
%
}
\label{fig:consgen}
\end{figure}

\section{Algorithm}\label{sec:algorithm}

Next, we describe our algorithm for solving the
constraints generated by $\conssym$, \ie for
for finding a satisfying assignment -- and hence
inferring suitable refinement types.
Our algorithm is an instance of the classical
``unfolding'' method for logic
programs~\cite{burstall-darlington-77,pettorosi94,TamakiSato84}.
We show how to apply unfolding to the NNF
constraints derived from refinement typing,
in a \emph{scoped} fashion that avoids the
blowup caused by long let-chains.
First, we describe a procedure for computing
the \emph{scope} of a $\kvar$
variable~(\S~\ref{sec:algo:scope}).
Second, we show how to use the scope to compute
the \emph{strongest solution} for a single $\kvar$
variable~(\S~\ref{sec:algo:solk}).
Third, we describe how we can \emph{eliminate}
a $\kvar$ by replacing it with its strongest
solution~(\S~\ref{sec:algo:elimk}).
Fourth, we show how the above procedure can
be repeated when the constraints are
acyclic~(\S~\ref{sec:algo:acyclic}).
Finally, we describe how to use our method
when the constraints have
cycles~(\S~\ref{sec:algo:cyclic}).

\subsection{The Scope of a Variable} 
\label{sec:algo:scope}

\begin{figure}[t!]
$$\begin{array}{lcl}
\toprule
\scopesym & : & (\Kvar \times \Cstr) \rightarrow \Cstr \\
\midrule
\scoped{\kvar}{\csand{\cstr_1}{\cstr_2}} & & \\ 
\quad \spmid \kvar      \in \cstr_1, \kvar \not \in \cstr_2
  & \doteq & \scoped{\kvar}{\cstr_1}
  \\ 
\quad \spmid \kvar \not \in \cstr_1, \kvar      \in \cstr_2
  & \doteq & \scoped{\kvar}{\cstr_2}
  \\ [0.05in]
\scoped{\kvar}{\csimp{x}{\base}{\pred}{\cstr'}} \\ 
\quad \spmid \kvar \not \in \pred
  & \doteq & \csimp{x}{\base}{\pred}{\scoped{\kvar}{\cstr'}}
  \\[0.05in]
\scoped{\kvar}{\cstr}
  & \doteq & \cstr
  \\[0.05in]
\bottomrule
\end{array}$$
\caption{The scope of a $\kvar$ in constraint $\cstr$.}
\label{fig:scoped}
\end{figure}


Procedure $\scoped{\kvar}{\cstr}$ (Fig.~\ref{fig:scoped})
returns a sub-constraint of $\cstr$ of the form:
${\csimpss{x_i}{\pred_i}{\cstr'}}$
such that
(1)~$\kvar$ does not occur in any $\pred_i$, and
(2)~\emph{all} occurrences of $\kvar$ in $\cstr$ occur in $\cstr'$.
As occurrences of $\kvar$ in $\cstr$ occur under
$\binds{x_i}{\pred_i}$, we can omit these binders
from the solution of $\kvar$, thereby avoiding the
blowup from unfolding let-chains described in
\S~\ref{sec:overview:blowup}.
%
%
Restricting unfolding to $\scoped{\kvar}{\cstr}$
does not affect satisfiability. The omitted hypotheses
are redundant as they are already present at all
\emph{uses} of $\kvar$.
This intuition is captured by the following
lemmas that follow by induction on $\cstr$.

\begin{lemma}[\textbf{Scoped-Definitions}]\label{lem:scoped-definitions}
$\defns{\cstr}{\kvar} \equiv \defns{\scoped{\kvar}{\cstr}}{\kvar}$
\end{lemma}

%

\begin{lemma}[\textbf{Scoped-Uses}]\label{lem:scoped-uses}
$\uses{\cstr}{\kvar} \equiv \uses{\scoped{\kvar}{\cstr}}{\kvar} \cup C$
for some flat constraints $C$ such that $\kvar \not \in C$.
\end{lemma}

\begin{lemma}[\textbf{Scoped-Satisfaction}]\label{lem:scoped-sat}
If $\satisfies{\soln}{\cstr}$ then $\satisfies{\soln}{\scoped{\kvar}{\cstr}}$.
\end{lemma}

\subsection{The Strongest Solution} 
\label{sec:algo:solk}

\begin{figure}[t!]
$$\begin{array}{lcl}
\toprule
\solksym & : & (\Kvar \times \Cstr) \rightarrow \Pred \\
\midrule

\solk{\kvar}{\csand{\cstr_1}{\cstr_2}}
  & \doteq & \solk{\kvar}{\cstr_1} \vee \solk{\kvar}{\cstr_2} \\[0.05in]

\solk{\kvar}{\csimp{x}{\base}{\pred}{\cstr}}
  & \doteq & \existsp{x}{\base}{\pred \wedge \solk{\kvar}{c} } \\[0.05in]

\solk{\kvar}{\kvapp{\kvar}{y}}
  & \doteq & \bigwedge_i x_i = y_i \quad \mbox{where} \ \params{x} = \args{\kvar} \\[0.05in]

\solk{\kvar}{\pred}
  & \doteq & \tfalse \\[0.05in]
\bottomrule
\end{array}$$
\caption{The Strongest Solution for a $\kvar$-Variable.}
\label{fig:solk}
\end{figure}

Procedure $\solk{\kvar}{\cstr}$ (Fig.~\ref{fig:solk})
returns a predicate that is guaranteed to satisfy all
the clauses where $\kvar$ appears as the head.
The procedure recursively traverses $\cstr$ to compute
the returned predicate as the \emph{disjunction} of
those bodies in $\cstr$ where $\kvar$ appears as the head.
When the head does not contain $\kvar$, the output predicate
is the empty disjunction, $\tfalse$.

\mypara{The Strongest Solution}
%
The \emph{strongest} (resp. \emph{scoped})
solution for $\kvar$ in $\cstr$, written
$\ssoln{\cstr}{\kvar}$ (resp. $\osoln{\cstr}{\kvar}$)
is:
\begin{align*}
\ssoln{\cstr}{\kvar} \ & \doteq\ \mapsingle{\kvar}{\lambda \params{x}. \solk{\kvar}{\cstr}} \\
\osoln{\cstr}{\kvar} \ & \doteq\ \mapsingle{\kvar}{\lambda \params{x}. \solk{\kvar}{\cstr'}}
  \ \mbox{where}\
    \scoped{\kvar}{\cstr} \equiv \csimpss{x_i}{\pred_i}{\cstr'} \mbox{s.t.}\ \kvar \not \in \pred_i
\end{align*}
The name ``strongest solution'' is justified by the
following theorems.
%
%
First, we prove that $\ssoln{\cstr}{\kvar}$ satisfies
the definitions of $\kvar$ in $\cstr$.
\begin{lemma}[\textbf{Strongest-Sat-Definitions}]\label{lem:defns} 
  $\satisfies{\ssoln{\cstr}{\kvar}}{\defns{\cstr}{\kvar}}$
\end{lemma}

As an example, consider the constraint:
${\cstr_5 \doteq \csimps{a}{p\ a}{\csimps{b}{q\ b}{ \csand{\kvar(a)}{\kvar(b)}}}}.$
Below, we show each sub-constraint $\cstr_i$,
and the corresponding solution returned by $\solk{\kvar}{\cstr_i}$:
$$\begin{array}{rclrcl}
\cstr_1 & \doteq & \kvar(a)                  & \qquad \solk{\kvar}{\cstr_1} & \doteq & x = a \\
\cstr_2 & \doteq & \kvar(b)                  & \qquad \solk{\kvar}{\cstr_2} & \doteq & x = b \\
\cstr_3 & \doteq & \csand{\cstr_1}{\cstr_2}  & \qquad \solk{\kvar}{\cstr_3} & \doteq & x = a \vee x = b \\
\cstr_4 & \doteq & \csimps{b}{q\ b}{\cstr_4} & \qquad \solk{\kvar}{\cstr_4} & \doteq & \exbind{b}{q\ b \wedge (x = a \vee x = b)} \\
\cstr_5 & \doteq & \csimps{a}{p\ a}{\cstr_4} & \qquad \solk{\kvar}{\cstr_5}   & \doteq & \exbind{a}{p\ a \wedge (\exbind{b}{q\ b \wedge (x = a \vee x = b)})}
\end{array}$$
Note that in essence,
$\ssoln{\cstr_5}{\kvar} \doteq \mapsingle{\kvar}{\lambda x. \solk{\kvar}{\cstr_5}}$
maps $\kvar$ to the disjunction of the two hypotheses (bodies)
under which $\kvar$ appears, thereby satisfying both implications,
and hence satisfying $\defns{\cstr_5}{\kvar}$.

\mypara{The Strongest Scoped Solution}
The addition of hypotheses -- the binders $\binds{x_i}{\pred_i}$
under which $\kvar$ always occurs -- preserves validity.
Thus, Lemma~\ref{lem:defns} implies that
the strongest scoped solution satisfies
the definitions of $\kvar$ in $\cstr$.
\begin{theorem}[\textbf{Strongest-Scoped-Sat-Definitions}]\label{thm:defns}
$\satisfies{\osoln{\cstr}{\kvar}}{\defns{\cstr}{\kvar}}$.
\end{theorem}
%
%
Furthermore, if there exists a solution $\soln$
that satisfies $\cstr$ then the composition of
$\soln$ and the strongest scoped solution also
satisfies $\cstr$.
\begin{theorem}[\textbf{Strongest-Scoped-Sat-Uses}]\label{thm:strongest}
If   $\satisfies{\soln}{\cstr}$
then $\satisfies{\comp{\soln}{\osoln{\cstr}{\kvar}}}{\cstr}$.
\end{theorem}
By Theorem~\ref{thm:defns} and Lemma~\ref{lem:partition}
it suffices to restrict attention to uses of $\kvar$,
\ie to prove that if
$\satisfies{\soln}{\uses{\cstr}{\kvar}}$
then
$\satisfies{\comp{\soln}{\osoln{\cstr}{\kvar}}}{\uses{\cstr}{\kvar}}$.
We show this via a key lemma that states that \emph{any}
solution $\soln$ that satisfies $\cstr$, must assign $\kvar$
to a predicate that is implied by, \ie weaker than,
$\solk{\kvar}{\cstr}$.

\begin{lemma}[\textbf{Strongest-Solution}]\label{lem:strong}
  If   $\satisfies{\soln,\binds{x_i}{\pred_i}}{\cstr}$
  then $\satisfies{\soln,\binds{x_i}{\pred_i}}{\apply{\ssoln{\cstr}{\kvar}}{\pred} \Rightarrow \pred}$.
\end{lemma}

We use the above lemma to conclude that the
hypotheses in $\uses{\cstr}{\kvar}$ are
\emph{strengthened} under
$\comp{\soln}{\osoln{\cstr}{\kvar}}$,
and hence, that each corresponding goal remains
valid under the composed assignment.

\subsection{Eliminating One Variable}
\label{sec:algo:elimk}

\begin{figure}[t!]
$$\begin{array}{lcl}
\toprule
\elimksym & : & (\Kvar \times \Cstr) \rightarrow \Cstr \\
\midrule
\elimk{\kvar}{\cstr}             & \doteq & \elimsol{\osoln{\cstr}{\kvar}}{\cstr} \\
\quad \mbox{where}               &        & \\
\quad \quad \osoln{\cstr}{\kvar} & =      & \mapsingle{\kvar}{\lambda \params{x}. \solk{\kvar}{\cstr'}} \\
\quad \quad \csimpss{x_i}{\pred_i}{\cstr'} & =      & \scoped{\kvar}{\cstr} \\
\quad \quad \params{x}           & =      & \args{\kvar}\\[0.05in]
\toprule
\elimsolsym & : & (\Soln \times \Cstr) \rightarrow \Cstr \\
\midrule
\elimsol{\soln}{\cstr_1 \wedge \cstr_2}
  & \doteq
  & \elimsol{\soln}{\cstr_1} \wedge \elimsol{\soln}{\cstr_2}
  \\[0.05in]
\elimsol{\soln}{\csimp{x}{\base}{\pred}{\cstr}}
  & \doteq
  & \csimp{x}{\base}{\apply{\soln}{\pred}}{\elimsol{\soln}{\cstr}}
  \\[0.05in]
\elimsol{\soln}{\kvapp{\kvar}{y}} &  &  \\[0.05in]
  \quad \spmid \kvar \in \domain{\soln}
  & \doteq
  & \ttrue
  \\[0.05in]
\elimsol{\soln}{\pred}
  & \doteq
  & \pred
  \\[0.05in]
\bottomrule
\end{array}$$
\caption{Eliminating Variables.}
\label{fig:elimsol}
\label{fig:elimk}
\end{figure}

Next, we show how to eliminate a single $\kvar$ variable
by replacing it with its strongest solution.

\mypara{Variable Elimination ($\elimsol{\soln}{\cstr}$)}
Procedure $\elimsol{\soln}{\cstr}$ (Fig.~\ref{fig:elimsol})
\emph{eliminates} from $\cstr$ all the $\kvar$ variables
defined in $\soln$ by replacing them with $\apply{\soln}{\kvar}$.
That is, the procedure returns as output a constraint $\cstr'$ where:
(1)~every \emph{body}-occurrence of a $\kvar$
   \ie where $\kvar$ appears as a hypothesis,
   is replaced by $\apply{\soln}{\kvar}$, and
(2)~every \emph{head}-occurrence of a $\kvar$
   \ie where $\kvar$ appears as a goal,
   is replaced by $\ttrue$ if $\kvar$
   is in the domain of $\soln$, and left
   unchanged otherwise.
   %
%
We prove by induction over $\cstr$ that
eliminating a single variable $\kvar$
yields a constraint over just the
\emph{uses} of that variable:

\begin{lemma}[\textbf{\elimsolsym}]\label{lem:subst}
If $\domain{\soln} = \{ \kvar \}$ and $\cstr' = \elimsol{\soln}{\cstr}$
then $\kvar \not \in \cstr'$ and
$\flatten{\cstr'} = \apply{\soln}{\uses{\cstr}{\kvar}}$.
\end{lemma}

\mypara{Eliminating One Variable} 
Procedure $\elimk{\kvar}{\cstr}$ (Fig.~\ref{fig:elimk})
eliminates $\kvar$ from a constraint $\cstr$ by invoking
$\elimsym$ on the strongest (scoped) solution for $\kvar$
in $\cstr$.
We prove that if $\kvar$ is acyclic in $\cstr$
then $\elimk{\kvar}{\cstr}$ returns an constraint
\emph{without} $\kvar$ that is
satisfiable if and only if $\cstr$ is satisfiable.

\begin{theorem}[\textbf{Elim-Preserves-Satisfiability}]
\label{thm:elim}
$\elimk{\kvar}{\cstr}$ is satisfiable iff $\cstr$ is satisfiable.
\end{theorem}

Intuitively, if $\satisfies{\soln}{\elimk{\kvar}{\cstr}}$
then, via Theorem~\ref{thm:strongest} and
Lemma~\ref{lem:sat-composition}, we have
$\comp{\soln}{\osoln{\cstr}{\kvar}}$
satisfies $\cstr$.
Dually, if $\satisfies{\soln}{\cstr}$ then we show
that $\satisfies{\soln}{\elimk{\kvar}{\cstr}}$.
See \S~\ref{sec:proofs} for details.

\subsection{Eliminating Acyclic Variables}
\label{sec:algo:acyclic}

When the constraint $\cstr$ is acyclic,
we can iteratively eliminate all the variables
in $\Kvar$ to obtain a VC whose validity
determines the satisfiability of $\cstr$,
as each elimination preserves acyclicity:
\begin{lemma}[\textbf{Elim-Acyclic}]\label{lem:acyclic-subst}
If $\Kvar'$ is acylic in $\cstr$, ${\dep{\soln} \subseteq \dep{\cstr}}$
then $\Kvar'$ is acyclic in $\elimsol{\soln}{\cstr}$.
\end{lemma}

\mypara{Eliminating Many Variables} 
Procedure $\elim{\Kvar}{\cstr}$ (Fig.~\ref{fig:elim})
eliminates a set of acyclic $\Kvar$ variables by iteratively
eliminating each variable via $\elimksym$.
Using Theorem~\ref{thm:elim}, Lemma~\ref{lem:subst} and
Lemma~\ref{lem:acyclic-subst} we show that the elimination
of \emph{multiple} acyclic $\kvar$-variables also preserves
satisfiability.

\begin{corollary}[\textbf{Exact-Satisfaction}] 
\label{thm:elims}
Let $\cstr' \doteq \elim{\Kvar}{\cstr}$.
\begin{enumerate}
  \item $\kvars{\cstr'} = \emptyset$, \ie $\cstr'$ is a VC, and
  \item $\cstr'$ is satisfiable iff $\cstr$ is satisfiable.
\end{enumerate}
\end{corollary}

\newcommand\bcsat{\textsc{BcSat}\xspace}

\mypara{Unavoidable Worst-Case Blowup}
Note that in the worst case, the $\elimsym$ procedure
can cause an exponential blowup in the constraint size.
In general some form of blowup is unavoidable as even
the Boolean version of the constraint satisfaction problem
is complete for exponential time.
A \emph{boolean constraint} is an acyclic
constraint where each variable is of type
$\tbool$.
The \emph{Boolean Constraint Satisfaction} (\bcsat)
problem is to decide whether a given boolean
constraint is satisfiable.

\begin{theorem}\label{thm:exptime}
  \bcsat is \exptime-complete.
\end{theorem}

The above is a corollary of \cite{GodefroidYannakakis13}
and the fact that reachability of (non-recursive) boolean
programs is equivalent to the satisfaction
of (acyclic) boolean constraints via the
correspondence of ~\cite{HMC}.
Thus, as validity checking is in \nptime, it is not possible
to always generate compact (polynomial sized) formulas,
assuming $\exptime \not = \nptime$.

\mypara{Avoiding Let-Chain Blowup}
Recall the example from~\S~\ref{sec:overview:blowup}.
Let $\cstr$ be the NNF constraint in (\ref{eq:cstr-blowup})
and consider the sequence of intermediate constraints
produced by eliminating $\kvar_1,\ldots,\kvar_n$
from $\cstr$:
\begin{align}
\cstr_1     & \doteq \cstr \notag \\
\cstr_{i+1} & \doteq \elimk{\kvar_i}{\cstr_i} \quad \mbox{for} \quad i \in 1,\ldots,n \notag \\
\intertext{At each step, the strongest solution for $\kvar_i$ is
computed from $\scoped{\kvar}{\cstr_i}$, \ie the part of the
constraint derived from the subterm where the binder $\kw{x}_i$
is in scope. Thus, for each $i \in 1,\ldots,n$ we get solutions}
\solk{\kvar_i}{\cstr_i} &  \doteq \exists \vv.\ \vv = \kw{x}_{i-1} \wedge z = \vv \label{eq:compact-sol} \\
\intertext{where $\args{\kvar_i} = z$. For brevity of exposition, we simplify the above to}
\solk{\kvar_i}{\cstr_i} &  \doteq z = \kw{x}_{i-1} \notag \\
\intertext{and hence, note that $\elim{[\kvar_1,\ldots,\kvar_n], \cstr}$ yields the linear-sized VC}
\cstr_{n+1} & \doteq
   \csimps{\kw{x}_0}{0 \leq \kw{x}_0}
                    {\csimps{\kw{x}_1}{\kw{x}_1 = \kw{x}_0 }
                                        { \ldots \Rightarrow \csimps{\kw{x}_n}{  \kw{x}_n = \kw{x}_{n-1}}
                                                                    {\csimps{\vv}{\vv = \kw{x}_n}{0 \leq \vv}}
                                        }
                    } \notag
\end{align}
that is easily proved valid by the SMT solver.

\subsection{Eliminating Cyclic Variables}
\label{sec:algo:cyclic}

When the constraint $\cstr$ is cyclic, the
satisfaction problem is undecidable via a
reduction from the problem of checking the
safety of \textsc{while}-programs~\cite{hornsurvey}.
Consequently, we can only compute over-approximate
or conservative solutions via abstract interpretation~\cite{HMC}.
Our approach is to
(1)~compute the set of variables that make the constraints cyclic
    (\ie whose absence would make the constraints acyclic),
(2)~eliminate all \emph{except} those variables, and
(3)~use predicate abstraction to solve the resulting constraint,
yielding a method that is much faster \emph{and} more
precise~(\S~\ref{sec:eval}), than ``global'' Liquid Typing~\cite{LiquidPLDI08}.

\mypara{Cut Variables}
A set of refinement variables $\Kvarc$ \emph{cuts} a
constraint $\cstr$ if $\Kvar - \Kvarc$ is acyclic in $\cstr$.
We cannot compute the \emph{minimum} set of cut variables
as this is the NP-Complete minimum feedback vertex set
problem~\cite{Karp72}.
Instead we use a greedy heuristic to compute $\Kvarc$.
We compute the strongest connected component (SCC) digraph
for $\depstar{\cstr}$, and iteratively remove $\kvar$-variables
(and recompute the digraph) until each SCC has a single
vertex \ie the graph is acyclic.
%
%
The reasoning for Lemma~\ref{thm:elims}
yields the following corollary:

\begin{corollary}\label{thm:elims:cyclic}
If $\Kvarc$ cuts $\cstr$
then $\elim{\Kvar - \Kvarc}{\cstr}$
returns a constraint $\cstr'$ such that
$\kvars{\cstr'} \subseteq \Kvarc$ and
$\cstr'$ is satisfiable iff $\cstr$ is satisfiable.
\end{corollary}

\begin{figure}[t!]
$$\begin{array}{lcl}
\toprule
\satsym & : & (C \times \setof{R}) \rightarrow \{ \textsc{True}, \textsc{False} \} \\
\midrule
\sat{\cstr}{\quals} & \doteq & \solve{\cstrs}{\quals} \\
  \quad \mbox{where}     &        & \\
  \quad \quad \cstr'     & =      & \elim{\Kvarc}{\cstr} \ \mbox{where $\Kvarc$ cuts $\cstr$} \\
  \quad \quad \cstrs     & =      & \flatten{\cstr'} \\[0.1in]
\toprule
\flattensym              &    :   & C \rightarrow \setof{C} \\
\midrule
\flatten{\ttrue}
  & \doteq & \emptyset
  \\[0.05in]
\flatten{\pred}
  & \doteq & \{ \pred \}
  \\[0.05in]
\flatten{\csand{\cstr}{\cstr'}}
  & \doteq & \flatten{\cstr} \cup \flatten{\cstr'}
  \\[0.05in]
\flatten{\csimp{x}{\base}{\pred}{\cstr}}
  & \doteq & \{ \csimp{x}{\base}{\pred}{\cstr'} \spmid \cstr' \in \flatten{\cstr} \}
  \\[0.1in]
\toprule
\elimsym & : & (\setof{\Kvar} \times \Cstr) \rightarrow \Cstr \\
\midrule
\elim{[]}{\cstr} & \doteq & \cstr \\[0.05in]
\elim{\kvar : \kvarl}{\cstr} & \doteq & \elim{\kvarl}{\elimk{\kvar}{\cstr}} \\[0.05in]
\bottomrule
\end{array}$$
\caption{Checking Constraint Satisfaction}
\label{fig:elim}
\label{fig:sat}
\label{fig:flatten}
\end{figure}

\mypara{Exact and Approximate Constraint Satisfaction}
Procedure $\sat{\cstr}{\quals}$ (Fig.~\ref{fig:sat})
computes a set of cut variables $\Kvarc$, and then
uses $\elimsym$ to compute \emph{exact}
solutions for non-cut variables,
and finally computes \emph{approximate}
solutions for the residual cut variables
left over after elimination, using the
$\solvesym$ procedure from~\cite{LiquidPLDI08}
which computes the satisfaction of (non-nested)
constrained Horn Clauses (CHC) (obtained
via $\flattensym$) using predicate
abstraction~\cite{GrafSaidi97,Houdini}.
%
%
\begin{theorem}[\textbf{Satisfaction}] \label{thm:sat}
  \leavevmode
\begin{enumerate}
  \item If $\sat{\cstr}{\quals}$ then $\cstr$ is satisfiable.
  \item If $\cstr$ is acyclic, then $\sat{\cstr}{\quals}$
       iff $\cstr$ is satisfiable.
\end{enumerate}
\end{theorem}
These results follow from
Corollary~\ref{thm:elims},
~\ref{thm:elims:cyclic},
the properties of $\solvesym$
(Theorem 2, \cite{LiquidPLDI08}),
and as $\solvesym$ reduces to
$\smtvalid{\cstr}$ when
$\kvars{\cstr}$ is empty.

\newcommand\TOTALLOC{15885}

\newcommand\datastruct{\toolfont{Data-Struct}}
\newcommand\vecalgs   {\toolfont{Vec-Algos}}

\newcommand\hname[1]{\textbf{#1}\xspace}
\newcommand\locname{\hname{Code}}
\newcommand\specname{\hname{Spec}}

\newcommand\filesname{\hname{Files}}
\newcommand\benchname{\hname{Benchmark}}
\newcommand\timename{\hname{Time}}
\newcommand\oldname{\toolfont{L}} 
\newcommand\newname{\toolfont{F}}
\newcommand\tymename{\hname{Time(s)}}
\newcommand\qualsname{\hname{Qualifiers}}

\section{Evaluation}\label{sec:eval}

We have implemented \toolname within
the \toolfont{LiquidHaskell} refinement
type checker~\cite{LiquidICFP14}.
It is used by default in the current version
\footnote{https://github.com/ucsd-progsys/liquidhaskell}.
We evaluate the \emph{speed} and
\emph{precision} of \toolname on
two sets of benchmarks from the
\toolfont{LiquidHaskell} project
totalling more than 12KLOC.
Our formalism based on call-by-value
evaluation is sound for Haskell as we
also simultaneously prove termination
for all potentially bottom-inhabited
terms, as described in~\cite{LiquidICFP14}.
Our results show that \toolname yields nearly
$2 \times$ \emph{speedups} for safety verification
benchmarks and, by synthesizing the most precise types
more than $10 \times$ faster, actually \emph{enables}
theorem proving.

\mypara{Safety Property Benchmarks}
The first set of benchmarks
are drawn from the Haskell standard libraries
and detailed in~\cite{realworldliquid}.
\datastruct\ comprises applicative
data structures like red-black trees, lists (\kw{Data.List}),
splay trees (\kw{Data.Set.Splay}), and binary search trees
(\kw{Data.Map.Base}); we verify
termination and structure specific invariants
like ordering and balance.
\toolfont{Vector-Algorithms} comprises
a suite of imperative (\ie monadic) array-based
sorting algorithms; we verify termination and
the correctness of array accesses.
\toolfont{Text} and \toolfont{Bytestring}
are the standard libraries for high-performance
unicode text and byte-array processing which
are implemented via low-level pointer arithmetic;
we verify termination, memory
safety and correctness properties
specified by the library API.

\mypara{Theorem Proving Benchmarks}
Recent work shows how refinement typing can convert
legacy languages like Haskell into \emph{proof assistants}
where ordinary programs can be used to write arbitrary
proofs of correctness about the ``deep specifications''
of other programs, and have the proofs checked via
refinement typing~\cite{ReflectionARXIV}.
%
%
The second set of benchmarks corresponds to a set of
programs corresponding to such proofs.
\toolfont{Arith} includes theorems about the \emph{growth}
of the fibonacci and ackermann functions;
\toolfont{Fold} includes theorems about the
\emph{universality} of traversals;
\toolfont{Monoid}, \toolfont{Functor},
\toolfont{Applicative} and \toolfont{Monad}
includes proofs of the respective
category-theoretic \emph{laws} for the
\kw{Maybe}, \kw{List}, and \kw{Id}
instances of the respective typeclasses;
and finally,
\toolfont{SatSolver} and \toolfont{Unification}
are fully verified implementations of the respective
algorithms from the \toolfont{Zombie} suite which,
absent SMT support, requires significantly more
local annotations (proof terms) from the
user~\cite{zombie1,zombie2}.

\mypara{Methodology}
For each benchmark, we compare the performance
of \toolfont{LiquidHaskell} using the (\oldname)
\emph{global} refinement inference from \cite{LiquidPLDI08},
and using our (\newname) \emph{local}
refinement algorithm.
We compare the amount of \tymename,
in seconds, it took to check each benchmark.
Each benchmark represents several Haskell
files (modules); we report \qualsname,
the average number of qualifiers (predicates)
that were required to synthesize the types needed
for verification per file.
%
%
In addition to the variables needed to eliminate cycles,
we aggressively mark all refinement variables appearing
in templates for ``top-level'' types as cut-variables
to ensure that simple refinements (over qualifiers)
are synthesized for such functions.
All benchmarks were run on a MacBook Pro with
a 2.2 GHz Intel Core i7 processor, using the
\toolfont{Z3} SMT solver for checking validity.
Table~\ref{table:eval} summarizes the results.

\mypara{Safety Verification Results}
Two points emerge from the safety verification
benchmarks.
First, for the larger benchmarks
(\toolfont{text} and \toolfont{bytestring}),
for which there already exist suitable qualifiers
permitting global inference, the new local \toolname
algorithm yields significant speedups -- often \emph{halving}
the time taken for verification. (Note this includes end-to-end
time, including name resolution, plain type checking \etc, and
so the actual speedup for just refinement checking is even greater.)
Second, even for these benchmarks, \toolname reduces (nearly halves)
the number of \emph{required} qualifiers. This makes the checker
significantly easier to use, as the programmer does not have
reason about which qualifiers to use for intermediate terms.
By synthesizing strongest refinements, \toolname removes a
key source of unpredictability and hard-to-diagnose false alarms.

\mypara{Theorem Proving Results}
The improvement is more stark for the theorem
proving benchmarks: most of them can only be
checked using local inference.
There are several reasons for this.
The proofs are made possible by heavy use of
polymorphic proof combinators.
As the specifications are much more
complicated, the combinators' type
variables must be (automatically)
instantiated with significantly
more complex refinements that
relate many program variables.
Thus, it is very difficult for
the user to even \emph{determine}
the relevant qualifiers.
Even if they could, the qualifiers
have many free variable (parameters)
which causes an exponential blowup
when matching against actual program
variables, making global, abstract
interpretation based refinement
synthesis impossible.
In contrast, since the theorem proving
benchmarks have almost no cyclic
dependencies, \toolname
makes short work of automatically
synthesizing the relevant refinement
instantiations, making complex
proofs possible.

\mypara{Comparison with other Tools}
We are not aware of any other tool
that scales up to these programs.
\fstar~\cite{fstar} requires local
annotations as described
in \S~\ref{sec:overview},
\toolfont{Mochi}~\cite{TerauchiPOPL13}
requires no annotations but may diverge,
and does not support uninterpreted
functions which precludes all of
our benchmarks.
Similarly, existing Horn Solvers like
\toolfont{$\mu$Z3} may diverge,
while
\toolfont{Eldarica}~\cite{eldarica15},
\toolfont{HSF}~\cite{RybalchenkoHSF}, and
\toolfont{Spacer}~\cite{spacer} do not
support uninterpreted functions.



\begin{table}
\centering
\caption{Comparing Liquid inference (L)~\cite{LiquidPLDI08}
with \toolname (F).
The number of modules (files) per benchmark is listed
in parentheses.
\locname is the total number of lines of non-comment
and non-whitespace lines of code and
\specname is the total number of lines of specification
(\ie top-level signatures), computed by \kw{sloccount}.
\qualsname is the number of qualifiers needed for
refinement inference; * means a false-positive,
\ie the benchmark could not be verified using
the qualifiers provided as the intermediate
terms have refinements not expressible using
the given qualifiers.
\tymename is the total time in seconds needed
to verify the benchmark (or to return a false
positive, for times with a *).}
\begin{tabular}{@{}lrrrrrrrrrr@{}}
\toprule
\benchname               & \locname & \specname  & \multicolumn{2}{c}{\qualsname} & \multicolumn{2}{c}{\tymename} \\
                         &          &            & \multicolumn{1}{c}{\oldname}
                                                     & \multicolumn{1}{c}{\newname}
                                                                 & \multicolumn{1}{c}{\oldname}
                                                                              & \multicolumn{1}{c}{\newname}    \\

\midrule

\datastruct (8)            & 1818 & 408 &  5 &  4 & 126 &  94 \\
\vecalgs   (11)            & 1252 & 279 &  4 &  4 &  78 &  61 \\
\toolfont{Bytestring}(11)  & 4811 & 726 & 18 & 11 & 233 & 136 \\
\toolfont{Text}(17)        & 3157 & 818 & 9  & 5 &  349 & 231 \\

\midrule

\toolfont{Arith}(2)        & 270 & 46  & *  & 0  & *63   & 5 \\
\toolfont{Fold}(1)         & 70  & 29  & 0  & 0  & 78   & 1 \\
\toolfont{Monoid}(2)       & 85  & 16  & 0  & 0  & 3    & 1 \\
\toolfont{Functor}(3)      & 137 & 28 & 0  & 0  & 55  & 3 \\
\toolfont{Applicative}(2)  & 146 & 36 & *  & 0  & *70  & 2 \\
\toolfont{Monad}(3)        & 180 & 42 & *  & 0  & *35  & 3 \\
\toolfont{Sat-Solver}(1)   & 98  & 31  & *  & 0  & *48   & 1 \\
\toolfont{Unification}(1)  & 139 & 53  & *  & 1  & *240  & 3 \\
\bottomrule
\end{tabular}
\vspace{0.2in}
\label{table:eval}
\end{table}

\section{Related Work}\label{sec:related}


\mypara{Floyd-Hoare Logics and Model Checking}
The process of eliminating refinement variables
to get a single verification condition (checked by SMT)
is analogous to the VC-generation step used in ESC-style
checkers~\cite{NelsonOppen,ESCJava}.
Indeed we use the term ``strongest refinement'' or ``solution''
for the intermediate types to highlight the analogy with
the notion of ``strongest postconditions'' from Floyd-Hoare
Logic.
While \toolname seeks to minimize the size of the
generated VC, unlike in the intra-procedural
case~\cite{FlanaganSaxe00, Leino05}, \exptime-completeness
means we cannot get compact VCs~(Theorem~\ref{thm:exptime}).
In the presence of loops, direct VC generation is
insufficient, and one must compute over-approximations.
\toolname's elimination procedure can be viewed as
a generalization of the \emph{large block encoding}
method~\cite{Beyer09} where over-approximation is
performed not at each instruction, but only at
``back edges'' in the control-flow graph.
Similarly, in program analyses it is common to
``inline'' the code for a procedure at a call site
to improve precision.
%
%
In \cite{dijkstramonad}, the authors
show how to generalize the notions of
composing weakest preconditions across
higher-order functions via the notion
of a \emph{Dijkstra Monad}.
\toolname's elimination procedure is
a way to systematically generalize
and lift the Floyd-Hoare notions
of strongest-postconditions
(dually, weakest-preconditions),
large block encodings, and procedure
inlining to the typed, higher order setting.
However, unlike Dijkstra Monads, \toolname
exploits the compositional structure of
types to locally synthesize precise
refinements (\ie invariants) in the
presence of polymorphic collections
and higher-order functions, to allow
checking examples like \kw{ex2}, \kw{ex3},
and \kw{ex4} (\S~\ref{sec:overview}).


%

\mypara{Local Type Inference}
\toolname performs a \emph{local inference} in that \emph{if}
the constraints are acyclic, then \toolname is able
to synthesize all intermediate refinement types exactly.
However, refinements render the problem (and our solution)
quite different than classical local typing~\cite{pierce-turner},
even with subtyping~\cite{odersky-local-01}.
First, even when all top-level (recursive) functions have
signatures, the constraints may get cycles, for example,
at instantiation sites for \kw{fold} functions (whose
polymorphic type variables must be instantiated with
the analogue of a ``loop invariant''.)
Second, our approach is orthogonal to bidirectional
type checking. Indeed, they can be (and in our
implementation, are) combined to yield a simpler
system of constraints, but we still need $\solkcsym$
and $\elimsym$ to synthesize the strongest refinements
\emph{relating} different program variables.

\mypara{Refinement Inference}
There are several other approaches to synthesizing refinements.
First, \cite{Knowles09} shows how existentials can be used
to type let-binders.
Second, \cite{GordonRefinement09} shows how a form of bidirectional
typing can be used to infer \emph{some} intermediate types.
Third, \cite{LiquidPLDI08} introduces the liquid typing framework
for synthesizing refinements via abstract interpretation.
Fourth, \cite{Polikarpova16} shows that liquid typing can be made
bidirectional and presents a new demand driven (``round trip'')
algorithm for doing the abstract interpretation (and also solving
for the weakest solution.)
However, none of the above approaches is able to handle the idiomatic
examples shown in ~\S~\ref{sec:overview}, or can only do so \emph{if}
given a suitable abstract domain (via templates).
We can try to infer such templates via abstraction
refinement~\cite{HMC,Kobayashi11,TerauchiPOPL13} but
that approach is notoriously unpredictable and prone
to diverging, especially in the presence of uninterpreted
functions which are ubiquitous in our examples.
Finally, \cite{JagannathanZhu15} shows how machine
learning over dynamic traces can be used to learn
refinements (in a generalization of the approach
pioneered by \cite{Ernst01}).
However, this needs closed programs (that can be run), which
can limit applicability to higher order functions.


\mypara{Horn Clauses}
%
%
Horn clauses have recently become a popular ``intermediate representation''
for verification problems \cite{hornsurvey}, as they can be used to encode
the proof rules for classical imperative Floyd-Hoare logics, and concurrent
programs \cite{RybalchenkoHSF} among others.
However, current Horn Clause solvers \eg \cite{Hoder12,RybalchenkoHSF,eldarica15}
are based on CEGAR and interpolation and hence, to quote a recent
survey \cite{hornsurvey}: ``mainly tuned for real and linear integer
arithmetic and Boolean domains'' rendering them unable to check any
of our benchmarks which make extensive use of uninterpreted functions.
%
%
%
Our work shows how to (1) algorithmically generate NNF clauses from typed,
higher-order programs, in a way that preserves scoping, (2) use an
optimized form of ``unfolding''~\cite{burstall-darlington-77,pettorosi94,TamakiSato84}
to synthesize the most precise type and (3) thereby, obtain a method
for improving the \emph{speed}, \emph{precision} and \emph{completeness}
of refinement type checking.



\begin{acks}                
%
We thank the anonymous referees,
Nadia Polikarpova, Eric Seidel and
Niki Vazou for their invaluable
feedback on earlier drafts of
this paper.
This material is based upon work supported by the
\grantsponsor{GS100000001}{National Science
Foundation}{http://dx.doi.org/10.13039/100000001}
under Grant Numbers
\grantnum{GS100000001}{CCF-1422471},
\grantnum{GS100000001}{CCF-1223850}, and
\grantnum{GS100000001}{CCF-1218344},
and a generous gift from Microsoft Research.
\end{acks}

\bibliography{sw}

\appendix

\section{Appendix: Proofs}
\label{sec:proofs}

We include below the proofs of the key theorems.

\begin{proof} (Lemma~\ref{lem:strong})
Assume that $\pred \equiv \kvapp{\kvar}{z}$ where $\params{z} = \args{\kvar}$
The other cases are trivial as
$\apply{\ssoln{\cstr}{\kvar}}{\pred} = \pred$,
and
$\apply{\soln}{\pred_1} \Rightarrow \pred_1$ and
$\apply{\soln}{\pred_2} \Rightarrow \pred_2$
implies
$\apply{\soln}{\csand{\pred_1}{\pred_2}} \Rightarrow \csand{\pred_1}{\pred_2}$.
%
%
The proof is by induction on the structure of $\cstr$.
\begin{alignat}{3}
\intertext{\quad \emph{Case: $\cstr \equiv \kvapp{\kvar}{y}$.}}
\bcoz{By definition of $\solk{\cstr}{\kvar}$}
  && \apply{\ssoln{\cstr}{\kvar}}{\kvapp{\kvar}{z}} & = \params{z} = \params{y} \label{eq:zeqy}\\
\bcoz{Assume}
  && \soln, \binds{x_i}{\pred_i} & \models \kvapp{\kvar}{y} \notag \\
\bcoz{Hence}
  && \soln, \binds{x_i}{\pred_i} & \models \params{z} = \params{y} \Rightarrow \kvapp{\kvar}{z} \notag \\
\bcoz{By~\ref{eq:zeqy}}
  && \soln, \binds{x_i}{\pred_i} & \models \apply{\ssoln{\cstr}{\kvar}}{\kvapp{\kvar}{z}} \Rightarrow \kvapp{\kvar}{z} \notag \\
\intertext{\quad \emph{Case: $\cstr \equiv \pred'$, such that $\kvar \not \in \pred'$.}}
\bcoz{By definition of validity}
  && \soln, \binds{x_i}{\pred_i} & \models \tfalse \Rightarrow \kvapp{\kvar}{z} \notag \\
\bcoz{By the definition of $\solk{\cstr}{\kvar}$}
  && \apply{\ssoln{\cstr}{\kvar}}{\kvapp{\kvar}{z}} & = \tfalse \notag \\
\bcoz{Hence,}
  && \soln, \binds{x_i}{\pred_i} & \models \apply{ \ssoln{\cstr}{\kvar}   }{\kvapp{\kvar}{z}} \Rightarrow \kvapp{\kvar}{z} \notag \\
\intertext{\quad \emph{Case: $\cstr \equiv \csand{\cstr_1}{\cstr_2}$.}}
\intertext{\quad \quad Let $\pred_j = \solk{\cstr_j}{\kvar}$ for $j = 1,2$}
\bcoz{Thus, by the definition of $\solk{\cstr}{\kvar}$}
  &&\apply{\ssoln{\cstr}{\kvar}}{\kvapp{\kvar}{z}} & = \pred_1 \vee \pred_2 \label{eq:def-or} \\
\bcoz{Assume, for $j = 1,2$}
  && \soln, \binds{x_i}{\pred_i} & \models \wedge_j \cstr_j \notag \\
\bcoz{That is, for $j = 1,2$}
  && \soln, \binds{x_i}{\pred_i} & \models \cstr_j \notag \\
\bcoz{By IH, for $j = 1,2$}
  && \soln, \binds{x_i}{\pred_i} & \models \pred_j \Rightarrow \kvapp{\kvar}{z} \notag \\
\bcoz{Hence,}
  && \soln, \binds{x_i}{\pred_i} & \models (\pred_1 \vee \pred_2) \Rightarrow \kvapp{\kvar}{z} \notag \\
\bcoz{Thus, by~\ref{eq:def-or},}
  && \soln, \binds{x_i}{\pred_i} & \models \apply{\ssoln{\cstr}{\kvar}}{\kvapp{\kvar}{z}} \Rightarrow \kvapp{\kvar}{z} \notag \\
\intertext{\quad \emph{Case: $\cstr \equiv \csimp{x}{\base}{\pred}{\cstr'}$.}}
\intertext{\quad \quad Let $\pred' = \solk{\cstr'}{\kvar}$.}
\bcoz{Thus, by the definition of $\solk{\cstr}{\kvar}$}
  && \apply{\ssoln{\cstr'}{\kvar}}{\kvapp{\kvar}{z}} & = \pred'                                  \label{eq:def-imp-1} \\
  && \apply{\ssoln{\cstr}{\kvar}}{\kvapp{\kvar}{z}}  & = \existsp{x}{\base}{\pred \wedge \pred'} \label{eq:def-imp-2} \\
\bcoz{Assume}
  && \soln, \binds{x_i}{\pred_i} & \models \csimp{x}{\base}{\pred}{\cstr'} \notag \\
\bcoz{By \rulename{Sat-Ext}}
  && \soln, \binds{x_i}{\pred_i}, \bind{x}{\pred} & \models \cstr' \notag \\
\bcoz{By IH}
  && \soln, \binds{x_i}{\pred_i}, \bind{x}{\pred} & \models \apply{\ssoln{\cstr'}{\kvar}}{\kvapp{\kvar}{z}} \Rightarrow \kvapp{\kvar}{z} \notag \\
\bcoz{By~\ref{eq:def-imp-1}}
  && \soln, \binds{x_i}{\pred_i}, \bind{x}{\pred} & \models \pred' \Rightarrow \kvapp{\kvar}{z} \notag \\
\bcoz{By \rulename{Sat-Ext}}
  && \soln, \binds{x_i}{\pred_i} & \models \csimp{x}{\base}{\pred}{\pred' \Rightarrow \kvapp{\kvar}{z}} \notag \\
\bcoz{That is,}
  && \soln, \binds{x_i}{\pred_i} & \models \csimp{x}{\base}{(\pred \wedge \pred'}{\kvapp{\kvar}{z})} \notag \\
\bcoz{As $x \not \in \kvapp{\kvar}{z}$}
  && \soln, \binds{x_i}{\pred_i} & \models (\existsp{x}{\base}{\pred \wedge \pred'}) \Rightarrow \kvapp{\kvar}{z} \notag \\ 
\bcoz{By~\ref{eq:def-imp-2}}
  && \soln, \binds{x_i}{\pred_i} & \models  \apply{\ssoln{\cstr}{\kvar}}{\kvapp{\kvar}{z}} \Rightarrow \kvapp{\kvar}{z} \notag
\end{alignat}
\end{proof}


\begin{proof} (Theorem~\ref{lem:defns})
  The proof is an induction on the structure of $\cstr$.
  \begin{alignat}{3}
  \intertext{\quad \emph{Case: $\cstr \equiv \kvapp{\kvar}{y}$.}}
    \bcoz{By definition}
    && \defns{\cstr}{\kvar}        & \doteq \kvapp{\kvar}{y} \notag \\
    \bcoz{and}
    && \ssoln{\cstr}{\kvar}(\kvar) & \doteq \lambda \params{x}. \wedge_i x_i = y_i \quad \mbox{ where } \params{x} = \args{\kvar} \notag \\
    \bcoz{Hence,}
    && \apply{\ssoln{\cstr}{\kvar}}{\defns{\cstr}{\kvar}}
                                   & \doteq \wedge_i y_i = y_i \label{eq:tauto1}\\
    \bcoz{As~\ref{eq:tauto1} is a tautology} 
    && \ssoln{\cstr}{\kvar}        & \models \defns{\cstr}{\kvar} \notag \\
  \intertext{\quad \emph{Case: $\cstr \equiv \pred$.}}
    \bcoz{By definition}
    && \defns{\cstr}{\kvar}        & \doteq \ttrue \notag \\
    \bcoz{Hence,}
    && \apply{\ssoln{\cstr}{\kvar}}{\defns{\cstr}{\kvar}}
                                   & \doteq \ttrue \label{eq:tauto2}\\
    \bcoz{As~\ref{eq:tauto2} is a tautology} 
    && \ssoln{\cstr}{\kvar}        & \models \defns{\cstr}{\kvar} \notag \\
\intertext{\quad \emph{Case: $\cstr \equiv \csand{\cstr_1}{\cstr_2}$.}}
\intertext{\quad \quad
  For $i \in \{ 1, 2\}$, let
    $\pred_i = \solk{\kvar}{\cstr_i}$,
    $\ssoln{\cstr_i}{\kvar}(\kvar) =  \lambda \params{x}. \pred_i$.
  }
  \bcoz{By definition}
    && \defns{\cstr}{\kvar}          & \doteq \defns{\cstr_1}{\kvar} \cup \defns{\cstr_2}{\kvar} \label{eq:union} \\
  \bcoz{By IH}
    && \ssoln{\cstr_i}{\kvar}                       & \models \defns{\cstr_i}{\kvar} \label{eq:ih-1} \\
  \intertext{\quad Let $\cstr' \in \defns{\cstr_i}{\kvar}$. As $\kvar$ \emph{only} in head,
   $\cstr' \equiv \csimpss{y}{\pred}{\kvapp{\kvar}{x}}$.}
  \bcoz{By \ref{eq:ih-1}}
    && \binds{y}{\pred} & \models \apply{\ssoln{\cstr_i}{\kvar}}{\kvapp{\kvar}{x}} \notag \\
  \bcoz{By Lemma~\ref{lem:weakening}}
    && \binds{y}{\pred} & \models \apply{\ssoln{\cstr_1}{\kvar}}{\kvapp{\kvar}{x}} \vee \apply{\ssoln{\cstr_2}{\kvar}}{\kvapp{\kvar}{x}} \notag \\
  \bcoz{That is}
    && \binds{y}{\pred} & \models \apply{\ssoln{\cstr}{\kvar}}{\kvapp{\kvar}{x}} \notag \\
  \bcoz{So for any $\cstr' \in \defns{\cstr_i}{\kvar}$}
    && \ssoln{\cstr}{\kvar} & \models \cstr' \notag \\
  \bcoz{Hence}
    && \ssoln{\cstr}{\kvar} & \models \defns{\cstr_i}{\kvar} \notag \\
  \bcoz{and hence, by~\ref{eq:union}}
    && \ssoln{\cstr}{\kvar} & \models \defns{\cstr}{\kvar} \notag \\
\intertext{\quad \emph{Case: $\cstr \equiv \csimp{x}{\base}{\pred}{\cstr'}$.}}
\intertext{\quad \quad
  Let $\pred' = \solk{\kvar}{\cstr'}$,
      $\soln' = \mapsingle{\kvar}{\lambda \params{z}. \pred'}$, and
      $\ssoln{\cstr}{\kvar} = \mapsingle{\kvar}{\lambda \params{z}. \existsp{x}{\base}{\pred \wedge \pred'}}$,
  }
  \bcoz{By IH}
    && \soln' & \models \defns{\cstr'}{\kvar} \label{eq:models-inner} \\
\intertext{\quad \quad Consider any arbitrary $\cstr'' \in \defns{\cstr'}{\kvar}$.}
  \bcoz{As $\kvar\in\head{\cstr''}$}
    && \cstr'' & \doteq \csimps{x_1}{\pred_1}{ \ldots \Rightarrow \csimps{x_n}{\pred_n}{\kvapp{\kvar}{y}}} \label{eq:inner-def}\\
  \bcoz{By~\ref{eq:models-inner}}
    && \soln'  & \models \csimps{x_1}{\pred_1}{ \ldots \Rightarrow \csimps{x_n}{\pred_n}{\kvapp{\kvar}{y}}} \notag \\
  \bcoz{By \rulename{Sat-Base}}
    && []      & \models \csimps{x_1}{\pred_1}{ \ldots \Rightarrow \csimps{x_n}{\pred_n}{ \SUBST{\pred'}{\params{z}}{\params{y}} }} \notag \\
  \bcoz{By \rulename{Sat-Ext}}
    && \binds{x_i}{\pred_i} & \models \SUBST{\pred'}{\params{z}}{\params{y}}  \label{eq:imp-a}\\
  \bcoz{Trivially,} 
    && \bind{x}{\pred} & \models \pred \label{eq:imp-b}\\
  \bcoz{Hence, by~\ref{eq:imp-a},~\ref{eq:imp-b}} 
    && \bind{x}{\pred}, \binds{x_i}{\pred_i} & \models {\pred \wedge \SUBST{\pred'}{\params{z}}{\params{y}}} \notag \\
  \bcoz{and so by definition of validity}
    && \bind{x}{\pred}, \binds{x_i}{\pred_i} & \models \existsp{x}{\base}{\pred \wedge \SUBST{\pred'}{\params{z}}{\params{y}}} \notag \\
  \bcoz{by~\ref{eq:inner-def}}
    && \ssoln{\cstr}{\kvar} & \models \csimp{x}{\base}{\pred}{\cstr''} \notag 
\end{alignat}
  \quad \quad As the above holds for an arbitrary $\cstr'' \in \defns{\cstr'}{\kvar}$, we get $\satisfies{\ssoln{\cstr}{\kvar}}{\defns{\cstr}{\kvar}}$.
\end{proof}

\begin{proof} (Theorem~\ref{thm:defns})
\begin{alignat}{3}
\bcoz{Let}
  && \scoped{\kvar}{\cstr} & \doteq \csimpss{x}{\pred}{\cstr'} \mbox{s.t. $\kvar \not \in \params{p}$} \label{eq:def-scoped} \\
\bcoz{By definition}
  && \osoln{\cstr}{\kvar} & = \ssoln{\cstr'}{\kvar} \notag \\
\bcoz{By Theorem~\ref{lem:defns}}
  && \osoln{\cstr}{\kvar} & \models \defns{\cstr'}{\kvar} \notag \\
\bcoz{By definition of validity}
  && \osoln{\cstr}{\kvar} & \models \csimpss{x}{\pred}{ (\defns{\cstr'}{\kvar}) } \notag \\
\bcoz{By Lemma~\ref{lem:flat-defn}}
  && \osoln{\cstr}{\kvar} & \models \defns{(\csimpss{x}{\pred}{\cstr'})}{\kvar} \notag \\
\bcoz{By (\ref{eq:def-scoped}) and Lemma~\ref{lem:scoped-definitions}}
  && \osoln{\cstr}{\kvar} & \models \defns{\cstr}{\kvar} \notag
\end{alignat}
\end{proof}


\begin{proof} (Theorem~\ref{thm:strongest})
By Theorem~\ref{thm:defns} it suffices to prove that
if   $\satisfies{\soln}{\cstr}$
then $\satisfies{\comp{\soln}{\osoln{\cstr}{\kvar}}}{\uses{\cstr}{\kvar}}$.
\begin{alignat}{3}
\bcoz{Let}
  && \scoped{\kvar}{\cstr} & \doteq \csimpss{y}{q}{\cstr'}
     \mbox{s.t. $\kvar \not \in \params{q}$ and $\osoln{\cstr}{\kvar} = \ssoln{\cstr'}{\kvar}$}
     \label{eq:def-scoped-1}\\
\bcoz{Assume}
  && \soln & \models \cstr \label{eq:use-sat-assm} \\
\bcoz{By (\ref{eq:def-scoped-1}),  Lemma~\ref{lem:scoped-sat}}
  && \soln & \models \csimpss{y}{q}{\cstr'} \label{eq:sol-models-cstr-prime} \\
\bcoz{By Lemma~\ref{lem:scoped-uses}, (\ref{eq:def-scoped-1})}
  && \uses{\cstr}{\kvar} & = \csimpss{y}{q}{(\uses{\cstr'}{\kvar})} \cup C \quad \mbox{where $\kvar \not \in C$}  \label{eq:use-sat-cases} \\ 
\bcoz{Hence, by (\ref{eq:scoped-use-a}) and (\ref{eq:scoped-use-b})}
  && \comp{\soln}{\osoln{\cstr}{\kvar}} & \models \uses{\cstr}{\kvar} \notag \\
\intertext{\quad \emph{Case: \quad $\satisfies{\comp{\soln}{\osoln{\cstr}{\kvar}}}{\csimpss{y}{q}{(\uses{\cstr'}{\kvar})}}$} (\ref{eq:scoped-use-a})}
\bcoz{By (\ref{eq:use-sat-assm},\ref{eq:use-sat-cases}), Lemma~\ref{lem:partition}}
  && \soln & \models  \csimpss{y}{q}{(\uses{\cstr'}{\kvar})} \label{eq:asm-case-b} \\
\bcoz{Let $\cstr'' \in \uses{\cstr'}{\kvar}$ s.t.}
  && \cstr'' & \equiv \csimpss{x_i}{\pred_i}{\pred} \mbox{ and $\kvar \not \in \pred$} \label{eq:asm-cstr-def} \\
\bcoz{By (\ref{eq:asm-case-b})}
  && \soln & \models \csimpss{y}{q}{\csimpss{x_i}{\pred_i}{\pred}} \notag \\
\bcoz{By \rulename{Sat-Ext}}
  && \soln, \binds{y}{q} & \models \csimpss{x_i}{\pred_i}{\pred} \notag \\
\bcoz{By (\ref{eq:sol-models-cstr-prime})}
  && \soln,\binds{y}{q} & \models \cstr' \notag \\
\bcoz{By (\ref{eq:def-scoped-1}), Lemma~\ref{lem:strong}, for $i=1 \ldots n$}
  && \soln,\binds{y}{q} & \models \apply{\ssoln{\cstr'}{\kvar}}{\pred_i} \Rightarrow \pred_i \notag \\
\bcoz{As $\osoln{\cstr}{\kvar} = \ssoln{\cstr'}{\kvar}$  (\ref{eq:def-scoped-1})}
  && \soln,\binds{y}{q} & \models \apply{\osoln{\cstr}{\kvar}}{\pred_i} \Rightarrow \pred_i \notag \\ 
\bcoz{By repeating Lemma~\ref{lem:strengthening} and $\kvar \not \in \pred$}
  && \soln,\binds{y}{q} & \models \apply{\osoln{\cstr}{\kvar}}{\csimpss{x_i}{\pred_i}{\pred}}  \notag \\
\bcoz{That is}
  && \soln,\binds{y}{q} & \models \apply{\osoln{\cstr}{\kvar}}{\cstr''} \notag \\
\bcoz{As by (\ref{eq:def-scoped-1}) $\kvar \not \in \params{q}$}
  && \soln & \models \apply{\osoln{\cstr}{\kvar}}{  \csimpss{y}{q}{\cstr''}} \notag \\
\bcoz{Thus, by (\ref{eq:asm-cstr-def})}
  && \soln & \models \apply{\osoln{\cstr}{\kvar}}{  \csimpss{y}{q}{\uses{\cstr'}{\kvar}}} \notag \\
\bcoz{And by Lemma~\ref{lem:valid-composition}}
  && \comp{\soln}{\osoln{\cstr}{\kvar}} & \models \csimpss{y}{q}{\uses{\cstr'}{\kvar}} \label{eq:scoped-use-a} \\
\intertext{\quad \emph{Case: \quad $\satisfies{\comp{\soln}{\osoln{\cstr}{\kvar}}}{C}$} (\ref{eq:scoped-use-b})}
\bcoz{By (\ref{eq:use-sat-assm}, \ref{eq:use-sat-cases}), Lemma~\ref{lem:partition}}
  && \soln & \models C \label{eq:soln-sat-c} \\
\bcoz{As $\kvar \not \in C$, $\domain{\osoln{\cstr}{\kvar}} = \{ \kvar \}$}
  && \apply{\osoln{\cstr}{\kvar}}{C} & = C \notag \\
\bcoz{By (\ref{eq:soln-sat-c})}
  && \soln & \models \apply{\osoln{\cstr}{\kvar}}{C} \notag \\
\bcoz{By Lemma~\ref{lem:valid-composition}}
  && \comp{\soln}{\osoln{\cstr}{\kvar}} & \models C \label{eq:scoped-use-b}
\end{alignat}
\end{proof}

\begin{lemma}[\textbf{Eliminate-Definitions}]\label{lem:elim-defs} 
  $\flatten{\elimk{\kvar}{\cstr}} = \apply{\osoln{\cstr}{\kvar}}{\uses{\cstr}{\kvar}}$.
\end{lemma}

\begin{proof} (Lemma~\ref{lem:elim-defs})
The above follows from the observations that:
(a)~$\domain{\osoln{\cstr}{\kvar}} = \{ \kvar \}$,
(b)~$\elimk{\kvar}{\cstr} = \elimsol{\osoln{\cstr}{\kvar}}{\cstr}$, and
(c)~Lemma~\ref{lem:subst}, as $\domain{\osoln{\cstr}{\kvar}} = \{ \kvar \}$.
\end{proof}

\begin{proof} (Theorem~\ref{thm:elim})
Let $\cstr' = \elimk{\kvar}{\cstr} = \elimsol{\osoln{\cstr}{\kvar}}{\cstr}$.
We prove the two directions separately.

\smallskip

\emph{Case: $\cstr'$ is satisfiable implies $\cstr$ is satisfiable}
\begin{alignat}{3}
\bcoz{Assume $\cstr'$ is satisfiable with}
&& {\soln'} & \models {\cstr'} \notag \\
\bcoz{(by Lemma~\ref{lem:flatten})}
&& {\soln'} &  \models {\flatten{\cstr'}} \notag \\
\bcoz{(by Lemma~\ref{lem:elim-defs})}
&& {\soln'} &  \models {\apply{\osoln{\cstr}{\kvar}}{\uses{\cstr}{\kvar}}} \notag \\
\bcoz{(by Lemma~\ref{lem:valid-composition})}
&& {\comp{\soln'}{\osoln{\cstr}{\kvar}}} &  \models {\uses{\cstr}{\kvar}} \label{eq:sat-up} \\
\bcoz{(by Theorem~\ref{thm:defns}, Lemma~\ref{lem:sat-composition})}
&& {\comp{\soln'}{\osoln{\cstr}{\kvar}}} & \models {\defns{\cstr}{\kvar}} \label{eq:sat-dn} \\
\bcoz{By~\ref{eq:sat-up},~\ref{eq:sat-dn}, Lemma~\ref{lem:partition}}
&& {\comp{\soln'}{\osoln{\cstr}{\kvar}}} & \models \cstr   \quad \mbox{\ie~$\cstr$ is satisfiable.} \notag \\
\intertext{\quad \emph{Case: $\cstr$ is satisfiable implies $\cstr'$ is satisfiable}}
\bcoz{Assume $\cstr$ is satisfiable with}
&& {\soln} & \models {\cstr} \notag \\
\bcoz{by Theorem~\ref{thm:strongest}}
&& \comp{\soln}{\osoln{\cstr}{\kvar}} & \models \cstr  \notag \\
\bcoz{by Lemma~\ref{lem:partition}}
&& \comp{\soln}{\osoln{\cstr}{\kvar}} & \models \uses{\cstr}{\kvar}  \notag \\
\bcoz{by Lemma~\ref{lem:valid-composition}}
&& \soln & \models \apply{\osoln{\cstr}{\kvar}}{\uses{\cstr}{\kvar}}  \notag \\
\bcoz{by Lemma~\ref{lem:elim-defs}}
&& \soln & \models \flatten{\cstr'} \quad \mbox{\ie by Lemma~\ref{lem:flatten} $\cstr'$ is satisfiable.} \notag
\end{alignat}
\end{proof}

\end{document}